\documentclass[11pt]{article}
\usepackage[margin = 1in]{geometry}
\newcount\Comments  
\newcommand{\kibitz}[2]{\ifnum\Comments=1{\color{#1}{#2}}\fi}
\newcommand{\hma}[1]{\kibitz{blue}{[HONGYAO: #1]}}
\newcommand{\rmr}[1]{\kibitz{red}{[RESHEF: #1]}}
\newcommand{\dcp}[1]{\kibitz{cyan}{[DAVID: #1]}}

\usepackage{etex}
\RequirePackage{natbib}
\usepackage{url}
\usepackage{latexsym}
\usepackage{amssymb}
\usepackage{amsmath}
\usepackage{amsfonts}
\usepackage{bbm}
\usepackage{ifthen}
\usepackage{enumerate}
\usepackage{makeidx}
\usepackage{psfrag}
\usepackage{floatflt}
\usepackage{verbatim}
\usepackage{dsfont}
\usepackage{scalefnt}
\usepackage{epsfig}
\usepackage{epstopdf}
\usepackage{color} 
\usepackage{xcolor}
\usepackage{calc}
\usepackage{graphicx} 
\usepackage{subfig}
\usepackage{multicol}
\usepackage{tikz}
\usetikzlibrary{patterns,calc,intersections,through,backgrounds,decorations.pathreplacing}

\usepackage{amsthm}		
\usepackage{thmtools, thm-restate}

\usepackage{algorithmic}
\usepackage[linesnumbered,ruled,vlined]{algorithm2e}

\definecolor{auburn}{rgb}{0.43, 0.21, 0.1}
\colorlet{darkblue}{blue!35!black}
\usepackage[colorlinks,
	citecolor = darkblue,
	urlcolor = darkblue,
	linkcolor = darkblue]{hyperref}
	
\theoremstyle{plain}
\newtheorem{theorem}{Theorem}
\newtheorem{proposition}{Proposition}

\theoremstyle{definition}
\newtheorem{definition}{Definition} 

\newtheorem{lemma}{Lemma}
\newtheorem{claim}{Claim}

\newcommand{\alts}{A}			

\newcommand{\utilDomain}{U}		
\newcommand{\utilSpace}{\mathcal{U}}	

\newcommand{\qlDomain}{U_{QL}}
\newcommand{\qlSpace}{\mathcal{U}_{QL}}

\newcommand{\nqlDomain}{U^0}	
\newcommand{\nqlSpace}{\mathcal{U}^0}	

\newcommand{\plDomain}{U_\parallel}		
\newcommand{\plSpace}{\mathcal{U}_\parallel}

\newcommand{\uspaceintpl}{\tilde{\mathcal{U}}}
\newcommand{\udomainintpl}{\tilde{U}}

\newcommand{\favalt}{\bar{a}}
\newcommand{\worstalt}{\underline{a}}

\def\payment{z}					
\def\maxwp{p}					
 
\newcommand{\affcnst}{C}		
\newcommand{\affcoeff}{k}		

\newcommand{\toposet}{P}		

\newcommand{\GG}{G}
\newcommand{\HH}{H}
\newcommand{\FF}{F}

\newcommand{\setR}{\mathbb{R}}

\newcommand{\txtfor}{~\mathrm{for}~}
\newcommand{\txtst}{~\mathrm{s.t.}~}
\newcommand{\txtand}{~\mathrm{and}~}

\newcommand{\hsq}{\hspace{-0.1em}}	

\title{Social Choice with Non Quasi-linear Utilities%
\thanks{The authors thank Vincent Conitzer, Debasis Mishra, Aris Filos Ratsikas, Moshe Tennenholtz and the participants of IJCAI'16, NYCE'17 and IJCAI'17 for helpful comments. The authors also thank Brian Baisa for discussions on the impossibility of efficient mechanisms. Ma is supported by a Siebel scholarship. Meir is supported by the Israel Science Foundation (ISF) under Grant No. 773/16.}
}

\author{Hongyao Ma
\thanks{John A. Paulson School of Engineering and Applied Sciences, Harvard University, 33 Oxford Street, Maxwell Dworkin 242, Cambridge,
MA 02138, USA. Email: hongyaoma@seas.harvard.edu.} 
\and 
Reshef Meir
\thanks{Department of Industrial Engineering and Management,
Technion - Israel Institute of Technology, Bloomfield Building, Room 510, Technion City, Haifa 3200003, Israel. Email: reshefm@ie.technion.ac.il.}
\and
David C. Parkes
\thanks{John A. Paulson School of Engineering and Applied Sciences, Harvard University, 33 Oxford Street, Maxwell Dworkin 229, Cambridge,
MA 02138, USA. Email: parkes@eecs.harvard.edu.}
}

\date{}

\begin{document}

\maketitle

\begin{abstract}
Without monetary payments, the Gibbard-Satterthwaite theorem proves that under mild requirements all truthful social choice mechanisms must be dictatorships. When payments are allowed, the Vickrey-Clarke-Groves (VCG) mechanism implements the value-maximizing choice, and has many other good properties: it is strategy-proof, onto, deterministic, individually rational, and does not make positive transfers to the agents. By Roberts' theorem, with three or more alternatives, the weighted VCG mechanisms are essentially  unique for domains with quasi-linear utilities. The goal of this paper is to characterize domains of non-quasi-linear utilities where ``reasonable'' mechanisms (with VCG-like properties) exist. Our main result is a tight characterization of the maximal non quasi-linear utility domain, which we call the largest \emph{parallel domain}. We extend Roberts' theorem to parallel domains, and use the generalized theorem to prove two impossibility results. First, any reasonable mechanism must be dictatorial when the utility domain is quasi-linear together with any single non-parallel type. Second, for richer utility domains that still differ very slightly from quasi-linearity, every strategy-proof, onto and deterministic mechanism must be a dictatorship.
\end{abstract}

\section{Introduction} \label{sec:intro}

We study social choice mechanisms that aggregate individual preferences and select one among a finite set of alternatives. Our interest is in the existence of strategy-proof social choice mechanisms under general, non quasi-linear utility functions.
In the classical voting problem without money, the seminal Gibbard-Satterthwaite theorem~\citep{gibbard1973,satterthwaite1975} states that if agents' preferences can be any ordering over the alternatives, the only deterministic, onto (i.e. every alternative can be selected) and strategy-proof mechanisms for three or more alternatives are dictatorial. On the other hand, with the introduction of monetary transfers and quasi-linear utilities, the Vickrey-Clarke-Groves (VCG) mechanism~\citep{vickrey1961,clarke1971,groves1973} maximizes social welfare in dominant strategies, and can be generalized to implement any affine maximizer of agents' values~\citep{roberts1979,krishna1998}.

However, quasi-linearity is a strong assumption, violated for example in domains with budget constraints and problems with lack of liquidity~\citep{cramton1997fcc} or with wealth effect and risk-aversion~\citep{pratt1992risk}. Non quasi-linearity can also arise as a result of the timing of payments coupled with temporal preferences~\citep{Frederick2002}, when payments are contingent on agents' actions and the presence of payments affects decisions and thus the likelihood of contingencies~\citep{ma2016contingent,Ma_ijcai16}, and in the context of illiquid currencies such as points-based allocation schemes~\citep{kash2007optimizing}.

To the best of our knowledge, few papers considered social choice mechanisms with monetary transfers and non quasi-linear utilities.\footnote{A preliminary, short version of this paper~\citep{Ma_ijcai16_b} proved a special case of the impossibility result for two agents and the richest non quasi-linear utility domain, with additional assumptions of unanimity and neutrality. The proof techniques are different, and this preliminary version does not make use of the generalized Roberts theorem.}
The main result of this paper is a tight characterization of the maximal utility domain, which we name the \emph{largest parallel domain}, where there exist non-dictatorial mechanisms that are strategy-proof, onto, deterministic, individually rational and satisfy \emph{no subsidy} (i.e. no positive transfers from the mechanism to the agents). These properties are those of the VCG mechanism in quasi-linear utility domains.
As a special case, we prove that for utility domains that contain all quasi-linear types but do not reside within the largest parallel domain, the only mechanisms satisfying the above conditions must be dictatorial. 
The proofs make use of a generalized Roberts' theorem, which we extended to non quasi-linear utility domains with suitable properties. 
We also provide a negative result for a broader class of mechanisms: by allowing richer utility domains that still differ very slightly from  quasi-linearity, we establish the impossibility of non-dictatorial mechanisms even without requiring individual rationality or no subsidy.

A key observation is that the critical property of agents' utilities that enables non-dictatorial mechanisms is not the linear dependency on payments. 
We say the utility functions of an agent is of \emph{parallel type} if for any two alternatives $a$ and $b$, within the range of interest, no matter how much the agent is charged for $b$, to achieve the same utility, the additional amount she is willing to pay for $a$ stays the same. 
Quasi-linear utility functions have this property, but there can also be non quasi-linear parallel types. 
\hma{Updated using the languages that we used in the rebuttals}
Intuitively, a type being parallel requires that regardless of which alternative is selected, the agent's marginal cost for money is the same as long as she has the same utility level--- the trade-off with money depends on how happy the agent is, not how much she is paying. 
%
%
A domain where all types are parallel is called a {\em parallel domain}, and the largest parallel domain is the set of all parallel types.

The rest of the paper is organized as follows. After a brief discussion of related work, we provide in Section~\ref{sec:PD} a formal definition of parallel domains. 
We prove a positive result in Section~\ref{sec:gen_weighted_VCG}, that within the parallel domains, the family of {\em generalized weighted VCG mechanisms} are strategy-proof, onto, deterministic, individually rational (IR), and satisfy no subsidy. 
In Section~\ref{sec:roberts}, we generalize Roberts' Theorem~\citep{roberts1979} to parallel domains, and prove that when the differences in agents' willingness to pay for different alternatives are unrestricted, maximizers of affine functions of the willingness to pay are the only
implementable choice rules amongst mechanisms with these properties.
With this characterization, we prove in Section~\ref{sec:dictatorship} our main result--- that when
agents have types outside of the parallel domain, the only mechanisms that are strategy-proof, onto, deterministic, individually rational, and satisfy no subsidy must be fixed-price dictatorships, i.e.
there exists a dictator who chooses her favorite alternative given
fixed prices associated with each alternative.
We also develop a negative result for a broader class of mechanisms: by
allowing utility domains that are slightly richer in their non
quasi-linearity, we show that individual rationality and no subsidy can be relaxed, while the impossibility of non-dictatorial
mechanisms can still be established.  We give in
Section~\ref{sec:relaxing_P4P5}, for example, an impossibility result
for a {\em two-slopes domain},
when the utility function of each agent for each alternative is
linear, with the slope taking one of two possible values. Proofs omitted from the body of the paper are provided in Appendix~\ref{appx:proofs}.

\subsection{Related Work}

In social choice without monetary payments, the classical Gibbard-Satterthwaite theorem has been extended to more restricted preference
domains, including the saturated domains~\citep{kalai1979social}, linked domains~\citep{aswal2003dictatorial}, circular domains~\citep{sato2010circular}, and weakly connected domains~\citep{pramanik2015further}.  
Domains for which there are positive results have also been extensively studied, see Black's majority rule~\citep{black1948rationale}, Moulin's median voting schemes and generalizations~\citep{moulin1980strategy,nehring2007structure} and  results on graphs with metric spaces~\citep{barbera2001introduction,schummer2002,dokow2012}. 

For social choice with payments and quasi-linear utilities,
\citet{roberts1979} showed that with three or more
alternatives, when the values can take any real numbers, positive
association of differences is
necessary and sufficient for strategy-proof implementation, and that
the only implementable choice rules are affine maximizers of agent values. Such choice rules can be implemented by  weighted VCG
mechanisms~\citep{krishna1998}. Characterizations of strategy-proof implementations have also been developed for mechanism design problems in specific domains~\citep{myerson1981optimal,rochet1987necessary,lavi2003towards,bikhchandani2006weak,saks2005weak}. 
%

One approach to mechanism design without quasi-linearity
is to assume that the functional form of agent utility functions is
known to the designer, for example auctions
public budget constraints~\citep{dobzinski2012multi}, or
auctions with  known risk
preferences~\citep{maskin1984optimal}. 
In contrast, we assume that the functional form of agent utility functions are private.
For private non quasi-linear utilities, under suitable richness of type space, \citet{kazumuramechanism} prove a ``taxation principle" style characterization and a ``revenue uniqueness" result of truthful mechanisms, and show various applications to problems other than social choice, e.g. single item allocation. 
The existence of truthful and non-dictatorial social choice mechanisms is not discussed.

For the assignment problem with unit demand,  
the minimum Walrasian-equilibrium mechanism is known to be truthful~\citep{demange1985strategy,Alaei2010,aggarwal2009,dutting2009bidder,morimoto2015strategy},
 even for any general non-increasing utility function in payment.
On the other hand, truthfulness cannot be achieved together with Pareto-efficiency for allocation problems
in which agents may demand more than one unit of good, or when agents have multi-dimensional type spaces~\citep{kazumura2016efficiency,dobzinski2012multi,baisa2016efficient}.
We do not impose Pareto-efficiency (PE) in proving our impossibility results.\footnote{With PE, however, we prove a similar dictatorship  result, which in comparison to our main results, requires weaker assumptions on agents utility domains. See Theorem~\ref{thm:dictatorial_with_PE} in Appendix~\ref{appx:result_with_PE}.}
 Randomized mechanisms for bilateral trade~\citep{garratt2015efficient} and revenue-optimal auctions in very simple settings~\citep{baisa2017auction,che2012assigning,pai2014optimal} have
also been studied in the context of private budget constraints.
We focus here on social choice rather than assignment or allocation problems, settings for which there is more structure on agents' preferences and also indifference toward outcomes where an agent's own assignments are the same.

\section{Preliminaries} \label{sec:PD}

Denote $N = \{1, 2, \dots, n\}$ as the set of \emph{agents} and $A = \{a, b, \dots, m\}$ as the set of \emph{alternatives}. A \emph{social choice
  mechanism} accepts reports from agents as input, selects a single
alternative $a^\ast \in A$, and may also determine payments.
%
A mechanism is \emph{onto} if for any alternative $a \in \alts$, there exists a preference profile for which $a$ is selected. 
A mechanism is \emph{dominant-strategy incentive compatible} (DSIC)
if no agent can gain by reporting false preferences.   \hma{People complained that the space of all possible reports are not specified yet.}

We allow monetary transfers, and the utility of an agent may depend both on the selected alternative and her assigned payment. Denote $u_{i,a}(\payment)$ as the utility of agent $i \in N$ if alternative $a \in \alts$ is selected and she needs to pay $\payment \in \setR$. $u_i = (u_{i,a}, \dots, u_{i,m})$ determines agent $i$'s \emph{type} and is her private information. Denote $u = (u_1, \dots, u_n)$ as a \emph{type profile}, and $u_{-i} = (u_1, \dots, u_{i-1}, u_{i+1}, \dots, u_n)$ as the type profile of agents except for agent $i$.

Denote the utility of alternative $a$ to agent $i$ at zero payment as
$v_{i,a} \triangleq u_{i,a}(0)$, which we call the \emph{value} of
alternative $a$ to agent $i$. Let $\favalt_i \in \arg\max_{a\in
  \alts} v_{i,a}$ and $\worstalt_i \in \arg\min_{a \in \alts}
v_{i,a}$ be a most and a least preferred alternative at zero
payment. A utility profile $u$ is \emph{quasi-linear} if
$u_{i,a}(\payment) = v_{i,a} - \payment$ for all $i \in N$, $a \in
\alts$ and $\payment \in \setR$. In this case, the values $\{v_{i,a}
\}_{i \in N, a \in \alts}$ fully determine the type profile. Let
the quasi-linear domain $\qlDomain$ be the set of all quasi-linear
types of a single agent where the $v_{i,a}$'s can take any value in $\setR$, and let $\qlSpace \triangleq \prod_{i=1}^n \qlDomain$ be the set of all quasi-linear type profiles.

\medskip

We consider \emph{non quasi-linear utilities} such that for all $i\in N$ and all $a \in \alts$, \vspace{-0.5em}
\begin{enumerate}[({S}1)]
	\item $u_{i,a}(\payment)$ is continuous and strictly decreasing in $\payment$,
	\item  $\lim_{\payment \rightarrow +\infty}u_{i,a}(\payment) < \min_{a' \in \alts} v_{i,a'}$.
\end{enumerate} 

Property~(S1) guarantees that agents strictly prefer to make smaller payments. Property~(S2) means that every agent prefers the worst alternative at zero payment to any alternative at some very large payment.
Denote the general non quasi-linear utility domain $\nqlDomain$ as the set of all types of an agent satisfying (S1) and (S2), and let $\nqlSpace \triangleq \prod_{i=1}^n \nqlDomain$ be the general non quasi-linear utility domain for a set of $n$ agents.  %


A \emph{social choice mechanism} $(x,t)$ on type domain $\utilSpace \subseteq \nqlSpace$ is composed of a \emph{choice rule} $x: \utilSpace \rightarrow A$ and a \emph{payment rule} $t = (t_1, \dots, t_n):\utilSpace \rightarrow \mathbb R^n$. Thus if the reported type profile is $\hat{u} \in \utilSpace$, the choice made is $x(\hat{u})$, and the utility of agent $i$ is $u_{i,x(\hat{u})}(t_i(\hat{u}))$.
A mechanism $(x, t)$ is DSIC if and only if, for any agent $i\in N$, any type $u_i \in \utilDomain_i$ of agent $i$, and any reported profile from other agents $\hat{u}_{-i} \in \utilSpace_{-i}$, agent $i$ cannot gain by misreporting any type $\hat{u}_i \in \utilDomain_i$:
\begin{align}
	u_{i,x(u_i, \hat{u}_{-i})}(t_{i}(u_i, \hat{u}_{-i})) \geq u_{i,x(\hat{u}_i, \hat{u}_{-i})}(t_{i}(\hat{u}_i, \hat{u}_{-i})).
\end{align}
A mechanism is \emph{individually rational} (IR) if and only if, by truthfully participating in the mechanism, regardless of the reports made by the other agents, no agent can be worse off than having their worst alternative at zero payment selected and not making any payment.\footnote{For a mechanism where IR is violated, an agent may benefit from not participating. See Section~\ref{sec:dictatorship} for more discussions on \emph{voluntary participation}. Assuming for all $i \in N$ and $a \in \alts$, $v_{i,a}$ may take any non-negative value, for DSIC mechanisms, this definition of IR is equivalent to requiring that the utility of any truthful agent is non-negative.} That is, $\forall i \in N,  ~\forall u_i \in \utilDomain_i,  ~\forall \hat{u}_{-i} \in \utilSpace_{-i}$, 
\begin{align}
	u_{i,x(u_i, \hat{u}_{-i})}(t_{i}(u_i, \hat{u}_{-i})) \geq \min_{a \in \alts} v_{i,a}.
\end{align}

We are interested in mechanisms with the following set of properties. 
\vspace{-0.8em}
\begin{multicols}{2}
\begin{enumerate}[{P}1.]
  	\setlength\itemsep{0em}
	\item Dominant-strategy incentive compatible 
	\item Deterministic 
	\item Onto 
	\item Individually rational 	
	\item No subsidy
\end{enumerate} 
\end{multicols}
\vspace{-0.5em}
Ontoness only requires that all alternatives will be selected given some type profile, but does not require all payment schedules are achievable. No subsidy requires that the mechanism does not make positive transfers to the agents.


Before continuing, we review a well-known characterization
of deterministic DSIC mechanisms. 
We say that a mechanism is \emph{agent-independent} if 
an agent's payment is independent of her report, conditioned
on a particular alternative being selected; i.e. fixing the
type profile of the rest of the agents $u_{-i}$, $\forall u_i, u_i'
\in \utilDomain_i$, $x(u_i, u_{-i}) = x(u_i', u_{-i}) \Rightarrow
t_i(u_i, u_{-i}) = t_i(u_i', u_{-i})$. Given an agent-independent
mechanism and any $u_{-i} \in \utilSpace_{-i}$, if there exists $u_i \in \utilDomain_i$ s.t. $x(u_i, u_{-i}) = a$, let the {\em agent-independent price} be the payment $i$ pays when $a$ is selected: $t_{i,a}(u_{-i}) \triangleq t_i(u_i, u_{-i})$, which depends only on $u_{-i}$. Otherwise, if
$x(u_i, u_{-i}) \neq a$ for all $u_i \in \utilDomain_i$, let
$t_{i,a}(u_{-i}) \triangleq +\infty$. 
An agent-independent mechanism 
is also \emph{agent-maximizing} if given the agent-independent
prices $\{ t_{i,a}(u_{-i}) \}_{i \in N, a \in \alts}$, the alternative selected by the mechanism maximizes the
utilities of all agents simultaneously, i.e. $\forall u \in
\utilSpace$, $\exists a^\ast \in A$ s.t. $a^\ast \in \arg \max_{a \in
  A} u_{i,a}(t_{i,a}(u_{-i}))　$ for all $i \in N$, and $x(u) =
a^\ast$.
The properties of agent-independence and agent-maximization are necessary and sufficient for deterministic DSIC mechanisms with quasi-linear utilities~\citep{vohra2011}, and this equivalence can be easily generalized to general utilities that strictly decrease with payment. 
%
%

A dictatorship in social-choice without money identifies an agent $i^\ast$ as the dictator, and always selects her favorite alternative. We generalize this concept for social choice with money as follows:

\begin{definition}[Fixed Price Dictatorship] Under a fixed price dictatorship, there exists a dictator $i^\ast \in N$ and fixed prices $\vec{\payment} \in \setR^m$. Given any type profile $u \in \utilSpace$, one of the dictator's favorite alternatives under these prices is selected, i.e. $x(u) = a^\ast \in \arg\max_{a \in \alts} u_{i^\ast, a}(\payment_{a} )$, and the dictator pays $t_{i^\ast}(u) = \payment_{x(u)}$.

\end{definition}

\subsection{Parallel Domains} 

Given any type of an agent $u_i \in \nqlDomain$, for each alternative $a$, we define the \emph{willingness to pay} $\maxwp_{i,a}$ as the payment for $a$ at which the agent is indifferent between getting $a$ at this payment, and getting her least preferred alternative $\worstalt_i$ at zero payment:
\begin{align}
	\maxwp_{i,a} \triangleq u_{i,a}^{-1}(v_{i,\worstalt_i}). \label{equ:maxwp}
\end{align}
See Figure~\ref{fig:prl_domain}. $\maxwp_{i,a}$ is the maximum amount the agent can be charged if alternative $a$ is selected, without violating IR. $\maxwp_{i,\worstalt_i} = 0$ always holds, and (S1)-(S2) imply that for all $a \in \alts$, $\maxwp_{i,a}$ exists, and $0 \leq \maxwp_{i,a} < +\infty$.

\begin{figure}[t!]
\vspace{-0.0em}
\centering   

\begin{tikzpicture}[scale = 0.95][font = \normalsize]
\draw[->] (-0.2,0) -- (8.5,0) node[anchor=north] {$\payment$};

\draw[->] (0,-0.4) -- (0, 3.8) node[anchor=west] {$u_{i,a'}(\payment)$};


\draw  [-]  (-0.3, 3.7) to[out=-51, in=-200] (5, 0.8) to[out=-20, in=-190] (7.5, 0.2);
\draw  [dashed]   (-0.4, 2.1) to[out=-25, in=-205]  (0, 1.97) to[out=-25, in=-200] (3.2, 0.8) to[out=-20, in=-190] (5.5, -0.2);
\draw[dashdotted]  (-0.32, 1.2) parabola[bend at end] (2, -0.2);

\draw[dotted](1.8, 1.97) -- (2.3, 1.97) ;
\draw[dotted](1.8, 2.2) -- (1.8, -0.25);
\draw (1.8, -0.1) node[anchor=north] { {$u_{i,a}^{-1}(v_{i,b})$ } };

\draw[dotted](0, 0.8) -- (3.2, 0.8) ;
\draw[dotted](5, 0.8) -- (5.5, 0.8) ;
\draw [dotted](3.2, 1.1) -- (3.2, -0.3);
\draw (3.2, -0.2) node[anchor=north] { { $\maxwp_{i,b}$ }};

\draw [dotted](5, 1.1) -- (5, -0.3);
\draw (5, -0.2) node[anchor=north] { { $\maxwp_{i,a}$ }};

\draw (-0.05, 3.2)node[anchor = east] {$v_{i,a}$};
\draw (-0.05, 1.85)node[anchor = east] {$v_{i,b}$};
\draw (-0.0, 0.8)node[anchor = east] {$v_{i, \worstalt_i}$};

\draw[-] (6, 3) -- (6.5, 3) node[anchor=west] {$u_{i,\hspace{0.1em}a}\hspace{0.1em}(\payment)$};
\draw[dashed] (6, 2.4) -- (6.5, 2.4) node[anchor=west] {$u_{\hspace{0.1em}i,\hspace{0.1em}b\hspace{0.1em}}(\payment)$};
\draw[dashdotted] (6, 1.8) -- (6.5, 1.8) node[anchor=west] {$u_{i, \worstalt_i}(\payment)$};


\fill [pattern = north east lines, pattern color = black!20](0, 1.97) to[out=-25, in=-200] (3.2, 0.8) -- (5, 0.8) to[out=-25, in=-200](1.8, 1.97) -- (0, 1.97) ;	

\draw[dotted,thick,<->](0, 1.97) -- (1.8, 1.97) ;
\draw[dotted,thick,<->](1.8, 1.3) -- (3.6, 1.3) ;
\draw[dotted,thick,<->](3.2, 0.8) -- (5, 0.8) ;

\end{tikzpicture}
\caption{An example type in the parallel domain. All horizontal sections of the shaded area (e.g. the dotted arrows) are of the same length, i.e. utility curves are horizontal 
translations of each other within the range but need not
be straight lines. \label{fig:prl_domain}} 
\end{figure}
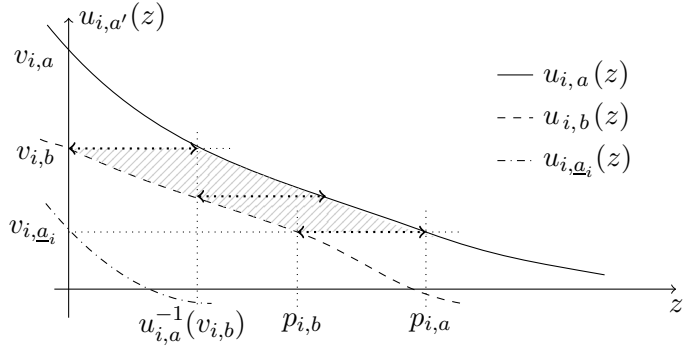

\begin{definition}[Parallel Domain] A utility domain $\utilDomain_i \subset \nqlDomain$ for an agent is a \emph{parallel domain} if for all $u_i \in \utilDomain_i$, 
\begin{align}
	u_{i,a} \left( \payment + (\maxwp_{i,a} \hspace{-0.1em}  - \maxwp_{i,b}) \right) = u_{i,b}(\payment), \: \forall a,b\in \alts \txtst v_{i,a} \hspace{-0.2em} \geq \hspace{-0.2em} v_{i,b},\: \forall \payment \hspace{-0.1em} \in \hspace{-0.2em} [0, ~\maxwp_{i,b}]. \label{equ:parallel_defn}
\end{align}
\end{definition}

See Figure~\ref{fig:prl_domain}. 
We call a $u_i \in \nqlDomain$ a \emph{parallel type} if \eqref{equ:parallel_defn} is satisfied. For a parallel type $u_i$, for $a,b \in \alts$ s.t. $v_{i,a} \geq v_{i,b}$, for any utility level $w \in [v_{i, \worstalt_i}, v_{i,b}]$, we have 
\begin{align}
	u_{i,a}^{-1}(w) - u_{i,b}^{-1}(w)  = \maxwp_{i,a} - \maxwp_{i,b} = u_{i,a}^{-1}(v_{i,b}). \label{equ:parallel_property}
\end{align}
In words, the differences in the payments on $a$ and $b$ in order to achieve $w$ is a constant that does not depend on $w$, i.e. the additional amount an agent is willing to pay for $a$ over $b$ does not depend on how much the agent is charged for $b$.%
\footnote{There is no requirement on the shape of the utility functions below the utility level $v_{i, \worstalt_i} = \min_{a \in \alts} v_{i,a}$, or where the payments are negative, since utilities in these ranges are irreverent to the incentive properties of mechanisms that are IR and satisfy no subsidy. This also implies that the shape of the utility function of the least preferred alternative is irrelevant, and that when there are only two alternatives the general non quasi-linear domain $\nqlDomain$ is parallel.}%
\footnote{Abstracting away the cardinal utility to an agent for an outcome, and modeling instead the prices at which one alternative is more preferable than another, a type being parallel then requires that  in the range of interest, for any pair of alternatives $a,b \in \alts$, which one is more preferable depends only on the difference in the prices $\payment_a - \payment_b$.
}
For an agent with parallel type, her marginal cost for money is the same as long as she has the same utility level, regardless of which alternative is selected. Intuitively, the trade-off with money depends on how happy the agent is, not how much she is paying. For example, while paying for a better alternative, an agent with this kind of wealth effect would care less about money at the same payment amount than compared to a weaker alternative.

Denote $\plDomain \subset \nqlDomain$ as the {\em largest parallel domain} (i.e. the set of all parallel types), and $\plSpace = \prod_{i=1}^n \plDomain$.
The quasi-linear domain $\utilSpace_{QL}$ is a parallel domain, where
$\maxwp_{i,a} = v_{i,a} -v_{i,\worstalt_i}$ for all $a \in \alts$ and
$\maxwp_{i,a} - \maxwp_{i,b} = v_{i,a} - v_{i,b}$ for all $a,b \in
\alts$. Another special case of the parallel domain is the
\emph{linear parallel domain}, where for every $u_i$, there exists
$\alpha > 0$ s.t. $\forall a \in \alts$ and all $\payment \in \setR$,
$u_{i,a}(\payment) = v_{i,a} - \alpha \payment$ (so the quasi-linear
domain is a special case when $\alpha=1$). For these two domains, the
``vertical" distances between the utility curves also stay the same (which need not be the case for a general parallel domain), and the utility functions are horizontal translations of each other
everywhere.

A utility domain for an agent $\utilDomain_i \subseteq \nqlDomain$ is
said to have \emph{unrestricted willingness to pay} if for any
$m$-dimensional non-negative vector with at least one zero entry,
there exists $u_i \in \utilDomain_i$ s.t. the willingness to pay according to $u_i$ is equal to this vector element-wise
(at least one zero entry is required since an agent always has zero willingness to pay for $\worstalt_i$).
Formally, $\forall \lambda \in \mathbb{R}_{\geq 0}^m$ for which $\exists a \in \alts$ s.t. $\lambda_a = 0$, there exists $ u_i \in \utilDomain_i$ s.t. $\maxwp_{i,a} =
\lambda_a$ for all $a \in A$.\footnote{There are multiple (actually,
infinite number of) types in $\nqlDomain$ with the same vector of
willingness to pay, and we only require at least one of them to be
included.} We call a parallel domain with unrestricted willingness to
pay an {\em unrestricted parallel domain}.
%
%
In particular, $\qlDomain$ is an example of an unrestricted parallel domain.
A utility domain $\utilSpace = \prod_{i=1}^n \utilDomain_i$ is an {\em unrestricted parallel utility domain} if each of the $\utilDomain_i$ is unrestricted and parallel.

We now prove two lemmas. 
\begin{restatable}{lemma}{LemStandardPrices} \label{lem:aiprice_characterization} 

Let $(x,t)$ be a DSIC and deterministic social choice mechanism on a utility domain $\utilSpace \subseteq \nqlSpace$ with unrestricted willingness to pay.
The mechanism satisfies (P4) IR and (P5) No subsidy if and only if $\forall i \in N$ and $\forall u_{-i} \in \utilSpace_{-i}$, the agent-independent prices $\{t_{i,a}(u_{-i})\}_{a\in \alts}$ satisfy:
\begin{enumerate}[(i)]
	\item $t_{i,a}(u_{-i}) \geq 0$ for all alternatives $a \in \alts$,
	\item there exists an alternative $a \in \alts$ s.t. $t_{i,a}(u_{-i}) = 0$.
\end{enumerate}

\end{restatable}

Thus, the agent-independent prices any agent faces under a mechanism satisfying (P1)-(P5) must be \emph{standard}, i.e. the minimum price among all alternatives is zero.
We leave the proof of the lemma to Appendix~\ref{appx:aiprice_characterization}.
%
%
\begin{lemma} \label{lem:maxwp_as_values} For any parallel type $u_i \in \plDomain$ and any standard prices $\{t_{i,a}\}_{a \in \alts}$:
\begin{enumerate}[(i)]
	\item $\forall a,b \in \alts$ s.t. $0 \leq t_{i,a} \leq \maxwp_{i,a}$ and $0 \leq t_{i,b} \leq \maxwp_{i,b}$, 
	\begin{align*}
		\maxwp_{i,a} - t_{i,a} \geq \maxwp_{i,b} - t_{i,b} \Leftrightarrow u_{i,a}(t_{i,a}) \geq u_{i,b}(t_{i,b}). 
	\end{align*}	
	\item $\arg\max_{a \in \alts} \{ u_{i,a}(t_{i,a}) \} = \arg \max_{a \in \alts} \{ \maxwp_{i,a} - t_{i,a} \}$.
\end{enumerate}
\end{lemma}

\begin{proof} We first prove part (i). 
Assume w.l.o.g that $v_{i,a} \geq v_{i,b}$. We know from the monotonicity of $u_{i,a}$ and the definition of parallel domain \eqref{equ:parallel_defn} that
\begin{align*}
	 & \maxwp_{i,a} - t_{i,a} \geq \maxwp_{i,b} - t_{i,b}  \Leftrightarrow t_{i,a} \leq t_{i,b} + \maxwp_{i,a}- \maxwp_{i,b} \\	  
	\Leftrightarrow & u_{i,a}(t_{i,a}) \geq u_{i,a}(t_{i,b} + (\maxwp_{i,a} -  \maxwp_{i,b} ) )
	\Leftrightarrow u_{i,a}(t_{i,a}) \geq  u_{i,b}(t_{i,b}).
\end{align*}

For part (ii), observe that if the price for at least one of the alternatives is zero, the highest utility at the given prices among all alternatives $\max_{a \in \alts} \{ u_{i,a}(t_{i,a}) \}$ is at least $\min_{a\in \alts} u_{i,a}(0)  = \min_{a \in \alts} v_{i,a}$. Therefore, for any alternative $a \in \alts$ s.t. $t_{i,a} > \maxwp_{i,a}$, the alternative cannot be agent-maximizing. Among the alternatives s.t. $t_{i,a} \leq \maxwp_{i,a}$, the agent-maximizing alternative(s) coincides with the maximizer(s) of $p_{i,a} - t_{i,a}$, according to part (i).
\end{proof}

Thus, in a parallel domain, the agent-maximizing alternative given standard prices is the maximizer of the difference between the willingness to pay and the price: $\maxwp_{i,a} - t_{i,a}$.
%
As a result, the willingness to pay serves similar roles as values in the quasi-linear domain, and it is this connection that enables us to generalize Roberts' theorem to unrestricted parallel domains.

\section{The Generalized Weighted VCG Mechanism} \label{sec:gen_weighted_VCG}

We prove in this section a positive result, that in parallel domains,
the generalized weighted VCG mechanisms implement in dominant strategy
any affine maximizer of willingness to pay: $	x(u) \in \arg\max_{a \in \alts} \left\lbrace \sum_{i=1}^n \affcoeff_i
\maxwp_{i,a} + \affcnst_a \right\rbrace,$
%
for non-negative weights $\{ \affcoeff_1\}_{i \in N}$, and real constants $\{\affcnst_a\}_{a \in  \alts}$.
\begin{definition}[Generalized Weighted VCG]
\label{defn:generalized_weighted_VCG} The generalized weighted VCG mechanism, parametrized by non-negative weights $\{\affcoeff_i \}_{i \in N}$ and real constants $\{\affcnst_a\}_{a \in \alts}$, collects a type profile $\hat{u} = (\hat{u}_1, \dots, \hat{u}_n)$ from agents, and computes the willingness to pay $\hat{\maxwp}_{i,a} = \hat{u}_{i,a}^{-1}(\min_{a'\in \alts} \hat{v}_{i,a'})$.
It is defined as
\begin{itemize}
	\item Choice rule: $x(\hat{u}) = a^\ast, \text{ where } a^\ast  \in \arg\max_{a \in \alts} \left\lbrace \sum_{i \in N} \affcoeff_i \hat{\maxwp}_{i,a} +  \affcnst_a \right\rbrace$, breaking ties arbitrarily, independent to agents' reports.
	\item Payment rule: $t_i(\hat{u}) = 0$ for $i \in N$ s.t. $\affcoeff_i = 0$; for $i$ s.t. $\affcoeff_i \neq 0$:
	\begin{align}
		t_i(\hat{u})  = \frac{1}{\affcoeff_i} \left( \sum_{j \neq i} \affcoeff_{j} \hat{\maxwp}_{j, a^\ast_{-i}} + \affcnst_{a^\ast_{-i}} - \sum_{j \neq i} \affcoeff_{j} \hat{\maxwp}_{j, a^\ast} - \affcnst_{a^\ast} \right), \label{equ:gwvcg_payment_rule}
	\end{align}
	where 
	$a^\ast_{-i} \in \arg\max_{a \in \alts}  \{ \sum_{j \neq i} \affcoeff_{j}  \hat{\maxwp}_{j,a} +  \affcnst_a \}$.
\end{itemize}
\end{definition}

\begin{theorem}
\label{thm:pos_result_PD} With type domain $\utilSpace \subseteq \plSpace$, any non-negative coefficients $\{\affcoeff_i\}_{i \in N}$ and any  real constants $\{\affcnst_a\}_{a \in \alts}$, the generalized weighted VCG mechanism is DISC, IR and does not make positive transfers to the agents.
\end{theorem}

\begin{proof} We first consider an agent $i \in N$ s.t. $\affcoeff_i > 0$. Given $u_{-i}$, we can check that for any $u_i$ agent $i$ reports s.t. $x(u_i,u_{-i}) = a$, agent $i$'s agent-independent payment would be
\begin{align}
	t_{i,a}(u_{-i}) = \frac{1}{\affcoeff_i} \left( \sum_{j \neq i} \affcoeff_{j} \maxwp_{j, a^\ast_{-i}} + \affcnst_{a^\ast_{-i}} - \sum_{j \neq i} \affcoeff_{j} \maxwp_{j, a} - \affcnst_{a} \right). \label{equ:ai_prices_gwVCG}
\end{align}
For agent $i$ s.t. $\affcoeff_i = 0$, $\arg\max_{a \in \alts}  \{ \sum_{j \in N} \affcoeff_{j}  {\maxwp}_{j,a} +  \affcnst_a \}$ and $ \arg\max_{a \in \alts}  \{ \sum_{j \neq i} \affcoeff_{j}  {\maxwp}_{j,a} +  \affcnst_a \}$ coincide. 
No matter what agent $i$ reports, $a_{-i}^\ast = a^\ast$ is always selected and she does not pay anything, thus
\begin{align}
	t_{i,a_{-i}^\ast}(u_{-i}) = 0 \text{ and } t_{i,a}(u_{-i}) = +\infty \text{ for } a \neq a_{-i}^\ast. \label{equ:ai_prices_gwVCG_k0}
\end{align}
Since $a^\ast_{-i}$ is the maximizer of $\sum_{j \neq i} \affcoeff_{j} \maxwp_{j, a} + \affcnst_{a}$, all agent-independent prices are non-negative. Moreover, $a^\ast_{-i}$ has the minimum price among all alternatives, which is exactly zero: $\min_{a \in \alts}  \{ t_{i,a}(u_{-i})\} = t_{i,a^\ast_{-i}}(u_{-i}) = 0$. By Lemma~\ref{lem:aiprice_characterization}, we know that the prices are standard, and the mechanism satisfies (P4) IR and (P5) No subsidy if it is DSIC. What is left to prove is choosing $a^\ast$ at such agent-independent prices is agent-maximizing for all agents, which implies DSIC. From Lemma~\ref{lem:maxwp_as_values} we know that we only need to prove $a^\ast$ is the maximizer of $\maxwp_{i,a} - t_{i,a}(u_{-i})$ for all agents. This is immediate for agents with $\affcoeff_i = 0$. For an agent with $\affcoeff_i > 0$, for any alternative $a \in \alts$,  we have
\begin{align*}
	& \maxwp_{i,a^\ast} - t_{i,a^\ast}(u_{-i}) - (\maxwp_{i,a} - t_{i,a}(u_{-i})) \\
	= &  \maxwp_{i,a^\ast} - \frac{1}{\affcoeff_i} \left( \sum_{j \neq i} \affcoeff_{j} \maxwp_{j, a^\ast_{-i}} + \affcnst_{a^\ast_{-i}} - \sum_{j \neq i} \affcoeff_{j} \maxwp_{j, a^\ast} - \affcnst_{a^\ast} \right) \\
	& - \maxwp_{i,a} + \frac{1}{\affcoeff_i} \left( \sum_{j \neq i} \affcoeff_{j} \maxwp_{j, a^\ast_{-i}} + \affcnst_{a^\ast_{-i}} - \sum_{j \neq i} \affcoeff_{j} \maxwp_{j, a} - \affcnst_{a} \right) \\
	=&  \frac{1}{\affcoeff_i} \left(\sum_{j \in N} \affcoeff_{j} \maxwp_{j, a^\ast} + \affcnst_{a^\ast} - \sum_{j \in N} \affcoeff_{j} \maxwp_{j, a} - \affcnst_{a} \right) \geq 0,
\end{align*}
thus $a^\ast$ indeed maximizes $\maxwp_{i,a} - t_{i,a}(u_{-i})$.
\end{proof}

Note that when $\utilSpace \subseteq \qlSpace$, we have $\maxwp_{i,a}
= v_{i,a} -v_{i,\worstalt_i}$ for all $a \in \alts$, $i \in N$, and
maximizing an affine function of the willingness to pay is equivalent
to maximizing the same affine function of the values. Thus, this
mechanism coincides with the weighted VCG mechanisms when utilities are quasi-linear.
Ontoness is satisfied if $\affcoeff_i > 0$ for at least one agent and when the utility domain is unrestricted.





\section{Generalizing Roberts' Theorem} \label{sec:roberts}

With quasi-linear
utilities,~\citet{roberts1979} showed that with three or more alternatives, if each agent's value for
each alternative can be any real number, the choice rule of any social
choice mechanism that is~(P1) DISC, (P2) deterministic and (P3)
onto must be a maximizer of some affine function of agents' values.
With two additional conditions, (P4) IR and (P5) No subsidy, we generalize Roberts' theorem to the unrestricted parallel
domains. 

\begin{restatable}[Roberts' Theorem on Parallel Domains]{theorem}{ThmRobertsPD} \label{thm:roberts} With three or more alternatives and an unrestricted parallel utility domain $\utilSpace$,
 for every social choice mechanism satisfying (P1)-(P5), there exist non-negative weights $\affcoeff_1, \dots, \affcoeff_n$ (not all equal to zero) and real constants $\affcnst_1, \dots, \affcnst_m$ such that for all $u \in \utilSpace$, 
\begin{align*}
	x(u) \in \arg \max_{a \in \alts} \left\lbrace \sum_{i=1}^n \affcoeff_i \maxwp_{i,a} + \affcnst_a  \right\rbrace.
\end{align*}
%
\end{restatable}

We prove in Lemma~\ref{lem:wmon_ic} that the \emph{weak monotonicity} condition defined in terms of willingness to pay (which is eqivalent to the W-Mon condition in terms of values~\citep{bikhchandani2006weak} when utilities are quasi-linear, see Definition~\ref{defn:wmon}) is a necessary condition of incentive compatibility. 
Following the same steps as in the first proof of Roberts' theorem presented in~\citet{lavi2009two} while treating willingness to pay as ``values" in the proof, we conclude that affine maximizers of the willingness to pay are the only implementable choice rules (see Appendix~\ref{appx:proof_roberts} for the details).
The coefficients $\{ \affcoeff_i\}_{i\in N}$ cannot be all zero in order to satisfy (P3) Ontoness.

Regarding the requirements on the utility domain: in the proof of
Robert's theorem~\citep{lavi2009two}, the values can take any real numbers, whereas the
willingness to pay for a non quasi-linear type takes non-negative
values and one of them has to be exactly zero. This does not prevent
us from generalizing the proof, since what is necessary is that the
differences in the values $v_{i,a} - v_{i,b}$ for all $a, b \in \alts$ can be any real numbers. 
We get this from the unrestricted parallel domain.

\begin{definition}[Weak Monotonicity] \label{defn:wmon} Let $\utilSpace = \prod_{i=1}^n \utilDomain_i \subseteq \nqlSpace$ be a utility domain. A choice rule $x: \utilSpace \rightarrow \alts$ satisfies \emph{weak monotonicity} (W-Mon) if for all $u_{-i} \in \utilSpace_{-i}$ and all $u_i, u_i' \in \utilDomain_i$, 
\begin{align*}
	x(u_i, u_{-i}) = a,~ x(u_i', u_{-i}) = b \Rightarrow \maxwp_{i,b}' - \maxwp_{i,b} \geq \maxwp_{i,a}' - \maxwp_{i,a}. 
\end{align*}
\end{definition}

In words, W-Mon in a parallel domain means that if alternative $a$ is selected under $(u_i, u_{-i})$ and alternative $b$
is selected under $(u_i', u_{-i})$, the additional willingness to pay for $b$
comparing with $a$ according to $u_i'$, i.e. $\maxwp_{i,b}'
-\maxwp_{i,a}'$, must be at least as big as the additional willingness
to pay for $b$ comparing with $a$ according to $u_i$: $\maxwp_{i,b} -
\maxwp_{i,a}$. This is a generalization of the W-Mon condition in terms of values as defined in~\citep{bikhchandani2006weak}, and the two are equivalent when utilities are quasi-linear, in which case $\maxwp_{i,a} - \maxwp_{i,b}= v_{i,a} - v_{i,b}$ holds for all $a,b \in \alts$.
\begin{lemma} 
 \label{lem:wmon_ic} With any parallel utility domain $\utilSpace \subseteq \plSpace$, every social choice mechanism satisfying (P1), (P2), (P4) and (P5) must satisfy W-Mon.
\end{lemma}

\begin{proof} Consider two types $u_i$, $u_i'$  and a social choice mechanism $(x,t)$ s.t. $x(u) = a$ and $x(u_i', u_{-i}) = b$. We know from agent-maximization and Lemma~\ref{lem:maxwp_as_values} that facing prices $\{t_{i,a'}(u_{-i})\}_{a' \in \alts}$, alternative $a$ must be a maximizer of $\maxwp_{i,a'} - t_{i,a'}(u_{-i})$ according to $u_i$, and alternative $b$ must be a maximizer of $\maxwp'_{i,a'} - t_{i,a'}(u_{-i})$ according to $u_i'$. Thus, $\maxwp_{i,a} - t_{i,a}(u_{-i}) \geq \maxwp_{i,b} - t_{i,b}(u_{-i})$ and $\maxwp_{i,b}' - t_{i,b}(u_{-i}) \geq \maxwp_{i,a}' - t_{i,a}(u_{-i})$ must hold. Adding both sides of the  two inequalities we get $\maxwp_{i,a} +  \maxwp_{i,b}'  \geq \maxwp_{i,b} +\maxwp_{i,a}' \Rightarrow \maxwp_{i,b}' - \maxwp_{i,b} \geq  \maxwp_{i,a}' - \maxwp_{i,a}$.
\end{proof}

\section{Impossibility results} \label{sec:dictatorship}

We now state the main result in this paper.
\begin{restatable}[Dictatorship]{theorem}{ThmDictatorship} \label{thm:dictatorship} With at least three alternatives and a utility domain $\utilSpace = \prod_{i=1}^n \utilDomain_i$ s.t. 
\begin{enumerate}[(C1)]
	\item for each $i\in N$, $\utilDomain_i$ contains an unrestricted parallel domain,
	\item for at least $n-1$ agents, $\utilDomain_i \not\subset \plDomain$,
\end{enumerate}
the only social choice mechanisms that satisfy (P1)-(P5) are fixed price dictatorships.
\end{restatable}

The quasi-linear domain is unrestricted and parallel, thus a
special case of the theorem can be stated as: on any utility domain
containing $\qlSpace$, if the utility domains of at least $n-1$ agents
contain non-parallel types, the only mechanisms satisfying (P1)-(P5)
are fixed price dictatorships. We provide here an outline of the proof and leave the full version to Appendix~\ref{appx:proof_dictatorship}.

Given any mechanism $(x,t)$ with (P1)-(P5) on a utility domain $\utilSpace$ satisfying (C1) and (C2), we prove that the restriction of the mechanism on the parallel subspace $ \uspaceintpl \triangleq \utilSpace \cap \plSpace$ must also satisfy (P1)-(P5).
Theorem~\ref{thm:roberts} then implies that on $\uspaceintpl$, the choice rule $x$ must be the maximizer of some affine function of agents' willingness to pay.
Fixing the choice rule, agent-independence and agent-maximization determine the agent-independent prices up to a constant (when there is no tie), and the requirement that prices being standard fully pins down $\{t_{i,a}(u_{-i}) \}_{a \in \alts}$ as \eqref{equ:ai_prices_gwVCG} and \eqref{equ:ai_prices_gwVCG_k0} for all $i \in N$ and any $u_{-i}$ s.t. $u_{i'}$ is parallel for all $i'\neq i$. 
Theorem~\ref{thm:roberts} requires that there exists at least one agent with a non-zero coefficient $\affcoeff_i$. 
If the number of agents for whom $\affcoeff_i > 0$ is more than one,  $\utilDomain_i \not\subset \plDomain$ for at least one of them. Assume w.l.o.g. that $\affcoeff_1, \affcoeff_2 > 0$ and $\exists u_2^\ast \in \utilDomain_2 \backslash \plDomain$. Fixing some parallel type profile of the rest of the agents $u_{-1,-2}^\ast$, $\{t_{1,a}(u_2^\ast, u_{-1,-2}^\ast) \}_{a\in\alts}$ are not yet pinned down by the above mentioned characterization since $u_2^\ast$ is not parallel. We prove that no agent-independent price $\{t_{1,a}(u_2^\ast, u_{-1,-2}^\ast) \}_{a\in\alts}$ guarantees that an alternative that is agent-maximizing for all agents always exists for all economies $(u_1, u_2^\ast, u_{-1,-2}^\ast)$ where $u_1 \in \utilDomain_1$. This contradicts DSIC. 

Now we know that there exists exactly one agent (say $i^\ast$) with a non-zero coefficient. This implies that when the reported profile is parallel (i.e. $u \in \utilSpace \cap \plSpace$), the outcome of the mechanism must be determined according to a fixed price dictatorship. By induction on the number of agents whose type is not parallel, we then prove that for any $u \in \utilSpace$, the outcome must also be determined by the same fixed price dictatorship.

\paragraph{On the Tightness of the Negative Result}

That the utility domain $\utilSpace$ contains an unrestricted parallel  domain is necessary for Theorem~\ref{thm:roberts}. If the number of agents that has a type outside of the parallel domain is smaller than $n-1$, there are at least two agents whose types are always parallel. The generalized weighted VCG mechanism with $\affcoeff_i > 0$ only for these agents would satisfy all (P1)-(P5), and still would not be a
dictatorship.

Regarding the properties (P1)-(P5), the mechanism being (P1) DSIC and
(P2) deterministic are trivially necessary. (P3) ontoness is
required, since if some alternatives are never selected, the number of
alternatives can be effectively reduced to two, in which case all types
satisfying (S1) and (S2) are parallel, thus even for the most general
$\nqlSpace$, any generalized weighted VCG mechanism with coefficients
not all zero satisfies (P1)-(P5).

The conditions on the utility domain~(C1) and~(C2) in the statement of Theorem~\ref{thm:dictatorship} require only a small deviation from quasi-linearity or the parallel domain. As an example, with $n=1$ and $m=3$, the type domain $\utilDomain_1 = \qlDomain$ and $\utilDomain_2 = \qlDomain \cup \{ u_2^\ast \}$ for any $u_2^\ast \notin \plDomain$ (e.g. as illustrated in Figure~\ref{fig:exmp_nqltype_u2_ast}) satisfies both~(C1) and~(C2). 
Assuming only~(C1) and~(C2), truthful mechanisms that violate one or both of (P4) and (P5) may still exist. For the above described utility domain, the mechanism which always adds 1 to agent 2's payment but otherwise functions exactly as a generalized weighted VCG mechanism satisfies (P1)-(P3) and (P5). More detailed discussions are provided in Appendix~\ref{appx:relax_ir_np}, and we present an alternative impossibility result assuming only (P1)-(P3) in Section~\ref{sec:relaxing_P4P5}.
 
\begin{figure}[t!]

\centering   
\begin{tikzpicture}[scale = 0.9][font=\normalsize]
\draw[->] (-0.2,0) -- (4.5,0) node[anchor=north] {$\payment$};

\draw[->] (0,-0.4) -- (0, 3.8) node[anchor=west] {$u_{2,a'}(\payment)$};

\draw[-] (-0.2, 3.4) -- (3.5, -0.3);

\draw[dashed](-0.2, 2.5) -- (0,2.3) to [out=0, in=-180] (0.7,1.6)-- (2.6, -0.3);

\draw[dashdotted](-0.2, 1.5) -- (1.6, -0.3);

\draw[dotted](1, 2.5) -- (1.0, -0.2) node[anchor=north] { { \scriptsize $1$ } };

\draw[dotted](0, 0.3) -- (3.3, 0.3);

\draw[-] (2.4,3) -- (2.9, 3) node[anchor=west] {$u^\ast_{2,a}(\payment)$};
\draw[dashed] (2.4,2.3) -- (2.9, 2.3) node[anchor=west] {$u^\ast_{2,b}(\payment)$};
\draw[dashdotted] (2.4,1.6) -- (2.9, 1.6) node[anchor=west] {$u^\ast_{2,c}(\payment)$};

\draw (0, 3.1)node[anchor = east] {$v_{2,a}$};
\draw (0, 2.4)node[anchor = east] {$v_{2,b}$};
\draw (0, 1.3)node[anchor = east] {$v_{2,c}$};


\end{tikzpicture}
\caption{An example non-parallel type $u_2^\ast$ of agent 2. \label{fig:exmp_nqltype_u2_ast}}
\end{figure}
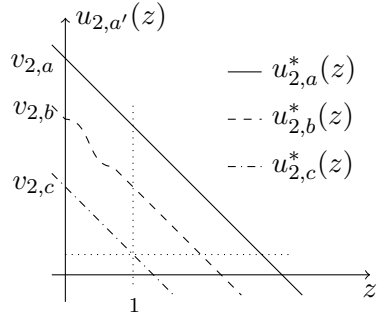

Note that Theorem~\ref{thm:dictatorial_with_PE} does not rely on any efficiency assumption. Replacing ontoness by Pareto-efficiency, we can prove a similar dictatorship result which instead of (C2), requires that only a single agent has a single non-parallel type. See Theorem~\ref{thm:dictatorial_with_PE} in Appendix~\ref{appx:result_with_PE}.

\paragraph{On the Fixed Price Dictatorship}

In a fixed price dictatorship, the dictator may still be asked to make
some non-zero payments. It is clear that such a mechanism is DSIC, and
when the prices that the dictator faces are standard, it is also IR
and never pays the agents. In order to 
replace the fixed-price dictatorship in Theorem~\ref{thm:dictatorship} with full dictatorship, i.e. 
the dictator chooses her favorite alternative free of charge, we can
impose another condition, \emph{voluntary participation} (VP), which
means any agent can choose to walk away from the mechanism and accept the
alternative decided by the rest of the agents without having to make
any payment.
If the dictator is charged a positive fixed price $\payment_a > 0$ for some alternative $a \in \alts$, and $a$ is still the dictator's favorite choice under the fixed prices, then when $a$ is selected by the sub-economy without the dictator (which should be the case for some $u_{-i^\ast}$ given onto-ness), the dictator would have the incentive to walk away, in which case $a$ will be selected and the dictator pays $0$. This contradicts VP.

VP is stronger than IR. To see this, note that IR requires that for every agent at least one
of the prices is weakly below the agent's willingness to pay.
On the other hand, VP
requires that the mechanism is well defined for any economy of $n-1$
agents and satisfy the same properties, moreover, each agent must face
a zero-price for the alternative that would be selected in the
sub-economy without her.

Regarding the payments from the rest of the agents--- when the dictator strictly prefers a single alternative $a \in \alts$ (i.e. $\forall a' \neq a$, $u_{i^\ast,a'}(\payment_{a'}) < u_{i^\ast,a}(\payment_{a})$), the rest of the agents do not make any payment: $t_i(u) = 0$ for all $i \neq i^\ast$. In the degenerate case where $|\arg \max_{a \in \alts} u_{i^\ast, a}(\payment_a) |> 1$, i.e. the dictator is indifferent toward multiple best alternatives, some or all the rest of the agents may be charged a non-zero payment to break ties among the dictator's favorite alternatives (e.g., we could run a generalized VCG mechanism between two alternatives that are tied from the dictator's perspective), and this may still satisfy (P1)-(P5). See Appendix~\ref{appx:proof_dictatorship}.

\subsection{Relaxing IR and No Subsidy} \label{sec:relaxing_P4P5}

We show in the rest of the section that with more richness in non quasi-linearity (e.g. with the linear domain with two slopes defined below), the dictatorship result remains given only (P1)-(P3).

\begin{definition}[Linear Domain with Two Slopes]
$\utilDomain_i$ is a linear domain with two slopes if there exists $ \alpha_i,~ \beta_i > 0$, $\alpha_i \neq \beta_i$ s.t. for all $u_i \in \utilDomain_i$, $\forall a \in \alts$, either $u_i(\payment) = v_{i,a} - \alpha_i \payment$ for all $\payment \in \setR$ or $u_i(\payment)  = v_{i,a} - \beta_i \payment$ for all $\forall \payment \in \setR$.  
\end{definition}

We say a linear domain with two slopes $\utilDomain_i$ is \emph{unrestricted} if for each $i \in N$, the values $\{v_{i,a}\}$ can be any real numbers, and the slopes of utility functions for different alternatives can be any combination of $\alpha_i$ and $\beta_i$.

\begin{restatable}{theorem}{ThmTwoSlopeDic} \label{thm:two_slope_dictatorship} With at least three alternatives and a utility domain $\utilSpace = \prod_{i=1}^n \utilDomain_i$ s.t. for each $i\in N$, $\utilDomain_i$ contains an unrestricted linear domain with two slopes, a social choice mechanism satisfying (P1)-(P3) must be a fixed price dictatorship.
\end{restatable}

Intuitively, each $\utilDomain_i$ contains as a sub-domain an unrestricted linear parallel domain (e.g. the set of $u_i$ s.t. $u_{i,a}(\payment) = v_{i,a} - \alpha_i \payment$ for all $z \in \setR$ and all $a\in \alts$), which is a special case of a \emph{strictly parallel domain}. For strictly parallel domains, with only (P1)-(P3), we can generalize Roberts' Theorem, and still determine the agent-independent prices up to a constant.
We then prove that if more than one agent has a positive coefficient in the choice rule, the induced agent-independent prices will result in contradictions when an agent's utility domain consists of utility functions with mixed slopes. Then we can prove by induction that the only agent with a positive coefficient must be a fixed price dictator. See Appendix~\ref{appx:two_slope_dictatorship} for the full proof.

The theorem still holds if $\alpha_i = \alpha$, $\beta_i = \beta$ for all $i$. Moreover, we would reach the same result if the $\alpha_i$'s and $\beta_i$'s are known to the mechanism designer. Since the $\alpha_i$'s and the $\beta_i$'s can be  arbitrarily close, this theorem shows that very slight disturbances on the slopes of agents' utility functions is sufficient to rule out the existence of truthful non-dictatorial mechanisms.

\section{Conclusions}

The existence of truthful and non-dictatorial social choice mechanisms strongly depends on whether monetary transfers are allowed. The seminal Gibbard-Satterthwaite theorem proves that without payments, the only truthful and onto mechanisms are dictatorial, whereas for the (rather restrictive) quasi-linear utility domain, any affine maximizer of agent values can be implemented if payments are allowed. 
%
%
We study social choice with payments and general utilities, distinguish  types being \emph{parallel} as the central property of quasi-linearity for DSIC mechanisms to exist, generalize (with additional conditions IR and No subsidy) Roberts' theorem to parallel domains with unrestricted willingness to pay, and provide a tight characterization of the largest parallel domain. Within the largest parallel domain, the generalized weighted VCG mechanisms implement any affine maximizer of agents' willingness to pay, and satisfy DSIC, onto, deterministic, IR and do not make payments to agents. Adding any non-parallel type to an unrestricted parallel domain, the only mechanisms with the above properties are dictatorial. 
We also discuss utility domains that are richer in their non quasi-linearity but still deviate very slightly from the quasi-linear domain, for which individual rationality and no subsidy can be relaxed and the dictatorship result remains.

\newpage

\bibliography{nqlsc-ecma-refs-original}

\newpage
\appendix

\section{Proofs} \label{appx:proofs}

\subsection{Proof of Lemma~\ref{lem:aiprice_characterization}} \label{appx:aiprice_characterization}

\LemStandardPrices*

\begin{proof}
We first prove parts (i) and (ii) given (P4) and (P5). Assume part (i) does not hold, that there exists $u_{-i}$ and $a \in \alts$ s.t. $t_{i,a}(u_{-i}) < 0$. Consider $u_1$ s.t. $a = \favalt_i$. In any agent-maximizing mechanism, the agent is guaranteed utility at least $u_{i,a}(t_{i,a}(u_{-i})) > \max_{a' \in \alts} v_{i,a'}$, which is not possible without the mechanism making a positive payment to the agent, and this violates (P5). Now we only need to prove that it cannot be the case that $t_{i,a}(u_{-i}) > 0$ for all $a \in \alts$. This is obvious, since if otherwise, there exists $u_i \in \utilDomain_i$ for whom $p_{i,a} < t_{i,a}(u_{-i})$ for all $a\in \alts$ thus $\max_{a \in \alts} u_{i,a}( t_{i,a}(u_{-i}) ) < \min_{a\in \alts} v_{i,a}$, in which case (P4) is violated.

We now prove the other direction. Given part (i), it is obvious that no matter which alternative is selected, the transfer from any agent to the mechanism is non-negative thus (P5) holds. Given (ii), we know that the agent's minimum possible utility under an agent-maximizing mechanism is $v_{i,a}$, which is at least $\min_{a' \in \alts} v_{i,a'}$, thus (P4) holds. 
\end{proof}

\subsection{Proof  of Theorem~\ref{thm:roberts}} \label{appx:proof_roberts}

\ThmRobertsPD*

We had proved in Lemma~\ref{lem:wmon_ic} that W-Mon is a necessary condition for (P1)-(P5) if the utility domain is unrestricted and parallel. We provide here a few more steps where the details differ slightly for values and willingness to pay, following the first proof presented in \citet{lavi2009two}, however, the high level ideas are the same.

\begin{definition}[Positive Association of Differences] \label{defn:pad} A social choice function $x$ on a parallel domain $\utilSpace$ satisfies \emph{positive association of differences} (PAD) if: for all $u$ and $u'$ in $\utilSpace$, if $x(u) = a$ and $\maxwp'_{i,a} - \maxwp_{i,a} > \maxwp'_{i,b} - \maxwp_{i,b}$ for all $b \neq a$ and all $i \in N$, then it must be the case that $x(u') = a$, as well. 
\end{definition}

\begin{lemma}[IC $\Rightarrow$ PAD] \label{lem:pad_ic}
A social choice mechanism on an unrestricted parallel domain $\utilSpace$  with (P1), (P2), (P4) and (P5)  satisfies PAD.
\end{lemma}

\begin{proof} 

By Lemma~\ref{lem:wmon_ic}, $x$ must satisfy W-Mon. 
Let $u, u' \in \utilSpace$ be type profiles such that $x(u) = a$ and $\maxwp'_{i,a} - \maxwp_{i,a} > \maxwp'_{i,b} - \maxwp_{i,b}$ for all $b \neq a$ and all $i \in N$. We need to show that $x(u') = a$. 
Denote $u^{(\ell)} = (u'_1, \dots, u_\ell', u_{\ell+1}, \dots, u_n)$. We know $x(u^{(0)}) = x(u) = a$. Assume $x(u^{(\ell-1)}) = a$ for some $\ell > 0$, we show by contradiction that $x(u^{(\ell)}) = a$ must also hold. By induction, this implies that $x(u') = x(u^{(n)}) = a$. 

Assume that there exists $\ell \geq 0$ and $b \neq a$ s.t. $x(u^{(\ell-1)}) = a$ but $x(u^{(\ell)}) = b \neq a$. Since all players except player $\ell$ have the same type in $u^{(\ell-1)}$ and in $u^{(\ell)}$, we get by W-Mon that $\maxwp_{\ell,b}' - \maxwp_{\ell,b} \geq \maxwp_{\ell,a}' - \maxwp_{\ell,a}$, which contradicts the PAD assumption on $u'$ and $u$. Thus, $x(u^{(\ell)}) = a$ must hold. This completes the proof of the induction step.
\end{proof}

We now prove the result analogous to Claim~1 in \citet{lavi2009two}. 

\begin{claim}\label{clm:clm1} Assume a choice rule $x$ on an unrestricted parallel domain $\utilSpace$ satisfies PAD. Fix type profiles $u, u' \in \utilSpace$ s.t. $x(u') = a$. If $\maxwp'_{i,b} - \maxwp_{i,b} > \maxwp'_{i,a} - \maxwp_{i,a}$ holds for all $i \in N$ for some $b \in \alts$, then $x(u) \neq b$.
\end{claim}

\begin{proof} 
We follow the same proof as in~\citet{lavi2009two}, and prove this claim by contradiction.
The construction differs slightly from the original proof since the willingness to pay is always normalized s.t. the an agent's smallest willingness to pay among all alternatives is zero.


Suppose by contradiction that $x(u) = b$. For each $i \in N$, denote $\delta_i \triangleq \maxwp_{i,b}' - \maxwp_{1,b} - (\maxwp_{i,a}' - \maxwp_{i,a})$. We know that $\delta_i > 0$ for all $i \in N$, and in addition, $\maxwp_{i,a} - \maxwp_{i,a}' - \delta_i/2 = \maxwp_{i,b} - \maxwp_{i,b}' + \delta_i/2 > \maxwp_{i,b} - \maxwp_{i,b}'$. Consider $\tilde{\maxwp}$ s.t.  for all $i\in N$,
\begin{align*}
	\tilde{\maxwp_{i,a}} = & \maxwp_{i,a} - \delta_i /2, \\
	\tilde{\maxwp_{i,b}} = & \maxwp_{i,b},
\end{align*}
and for $c \in \alts$ s.t. $c \neq a, b$,
\begin{align*}
	\tilde{\maxwp}_{i,c} = \min \{ \maxwp_{i,c}, ~\maxwp_{i,c}' + \maxwp_{i,a} - \maxwp_{i,a}' \} - \delta_i.
\end{align*}
Then, let $\maxwp''$ be the normalized version of $\tilde{p}$, i.e.
\begin{align*}
	\maxwp_{i,c}'' \triangleq   \tilde{\maxwp_{i,c}} - \min_{d \in \alts} \{\tilde{\maxwp_{i,d}} \},  ~\forall i \in N,~\forall c \in \alts.
\end{align*}

First, we observe that $\min_{c \in \alts} \{ \maxwp_{i,c}''\} = 0$ holds for all $i \in N$, i.e. $\maxwp''$ is a valid set of willingness to pay for the agents. Since $\utilSpace$ is unrestricted, we can find $u'' \in \utilSpace$ where agents' willingness to pay is given by $\maxwp''$.
Second, observe that for all $i \in N$ and any pair of alternatives $c,d \in \alts$,  $\maxwp_{i,c}'' - \maxwp_{i,d}'' = \tilde{\maxwp_{i,c}} - \tilde{\maxwp_{i,d}}$. 

We now show that the PAD condition from $u \rightarrow u''$ results in $x(u'') = b$, whereas applying PAD condition from $u' \rightarrow u''$ results in $x(u'') = a$, thus a contraction. 

To show $x(u'') = b$ must hold, observe that for all $i \in N$,
\begin{align*}
	& \maxwp_{i,b}'' - \maxwp_{i,b} - (\maxwp_{i,a}'' - \maxwp_{i,a}) = \tilde{\maxwp_{i,b}} - \tilde{\maxwp_{i,a}}  - (\maxwp_{i,b} - \maxwp_{i,a}) \\
	= & \maxwp_{i,b} - \maxwp_{i,a} + \delta_i /2  - (\maxwp_{i,b} - \maxwp_{i,a})  = \delta_i /2 > 0,
\end{align*}
and that for all $c \neq a,b$,
\begin{align*}
	& \maxwp_{i,b}'' - \maxwp_{i,b} - (\maxwp_{i,c}'' - \maxwp_{i,c}) = \tilde{\maxwp_{i,b}} - \tilde{\maxwp_{i,c}}  - (\maxwp_{i,b} - \maxwp_{i,c}) \\
	\geq & \maxwp_{i,b} - \maxwp_{i,c} + \delta_i  - (\maxwp_{i,b} - \maxwp_{i,c}) = \delta_i > 0.
\end{align*}
From Lemma~\ref{lem:pad_ic}, we know that $x(u'') = b$ must hold. Similarly, to show $x(u'') = a$ must hold, we can check that for all $i \in N$,
\begin{align*}
	& \maxwp_{i,b}'' - \maxwp_{i,b}' - (\maxwp_{i,a}'' - \maxwp_{i,a}') = \tilde{\maxwp_{i,b}} - \tilde{\maxwp_{i,a}}  - (\maxwp_{i,b}' - \maxwp_{i,a}') \\
	= & \maxwp_{i,b} - \maxwp_{i,a} - (\maxwp_{i,b}' - \maxwp_{i,a}') + \delta_i /2  = -\delta_i/2 < 0,
\end{align*}
and that for all $c \neq a,b$,
\begin{align*}
	& \maxwp_{i,c}'' - \maxwp_{i,c}' - (\maxwp_{i,a}'' - \maxwp_{i,a}') = \tilde{\maxwp_{i,c}} - \tilde{\maxwp_{i,a}}  - (\maxwp_{i,c}' - \maxwp_{i,a}') \\
	\leq & \maxwp_{i,c}' + \maxwp_{i,a} - \maxwp_{i,a}' - \delta_i - (\maxwp_{i,a} - \delta_i /2) - (\maxwp_{i,c}' - \maxwp_{i,a}')  = -\delta_i/2 < 0.
\end{align*}
Lemma~\ref{lem:pad_ic} then implies $x(u'') = a$ must hold. 
\end{proof}

We now consider the same sets that are analyzed in~\citet{lavi2009two} in the first proof for Roberts' theorem, however, we define these sets in terms of differences in willingness to pay instead of differences in values. Let $\utilSpace$ be an unrestricted parallel domain. For all tuples $(a,b) \in \alts \times \alts$ that $a \neq b$:
\begin{align*}
	\toposet(a, b) \triangleq \{ \alpha \in \setR^n~|~\exists u \in \utilSpace \txtst p_a - p_b = \alpha \txtand x(u) = a \}
\end{align*}
The two immediate properties of the sets under quasi-linearity also holds:
\begin{enumerate}[1.]
	\item For every $a$ and $b$, the set $\toposet(a,b)$ is not empty, as long as the choice rule $x$ is onto.
	\item If $\alpha \in \toposet(a,b)$, then for any positive $\delta \in \setR^n$ s.t. $\delta_i > 0$ for all $i$, $\alpha + \delta \in \toposet(a,b)$. To see this, note that $\alpha \in \toposet(a,b)$ implies that there exists $u$ s.t. $x(u) = a$ and $\maxwp_a - \maxwp_b = \alpha$. Now we look for another type profile $u' \in \utilSpace$ s.t. comparing with $u$, we are increasing $\maxwp_a$ by $\delta$, while keeping the other willingness to pay the same. We know from PAD that $x(u') = a$ must still hold, and in this case $\maxwp_a' -\maxwp_b' = \alpha + \delta \in \toposet(a,b)$ as required. 
\end{enumerate}


We now prove the following claim, which is analogous to Claim~2 in \citet{lavi2009two}. Without otherwise specify, the claims in the rest of this section assumes that there are three or more alternatives, that the utility domain is parallel and unrestricted, and that the mechanism satisfies (P1)-(P5). 

\begin{claim} \label{clm:clm2}
For every $\alpha,\epsilon \in \setR^n$ , $\epsilon > 0$:
\begin{enumerate}[1.]
	\item $\alpha - \epsilon \in \toposet(a,b) \Rightarrow - \alpha \notin \toposet(b,a)$.
	\item $\alpha \notin \toposet(a,b) \Rightarrow -\alpha \in \toposet(a,b)$.
\end{enumerate}
\end{claim}

\begin{proof} For the first part, note that $\alpha - \epsilon \in \toposet(a,b)$ implies that there exists $u' \in \utilSpace$ s.t. $p'_a - p'_b = \alpha - \epsilon$ and $x(u') = a$. Now let $u \in \utilSpace $ be any type profile s.t. $\maxwp_b - \maxwp_a = -\alpha$. 
We know that for all $i \in N$, $\maxwp_{i,a} - \maxwp_{i,b} = \alpha_i > \alpha_i - \epsilon = \maxwp_{i,a}' - \maxwp_{i,b}'$. This implies 
$\maxwp_{i,b}' -\maxwp_{i,b}   > \maxwp_{i,a}'  - \maxwp_{i,a} $ for all $i \in N$, and we get $x(u) \neq b$ by applying Claim~\ref{clm:clm1}. This shows that for all $u \in \utilSpace$ s.t. $\maxwp_b - \maxwp_a = -\alpha$, $x(u) \neq b$, therefore $-\alpha \notin \toposet(b,a)$.

For the second part, for any $c \neq a,b$, take some $\beta_c \in \toposet(a,c)$ and fix some $\epsilon > 0$. Choose any $u$ s.t. $\maxwp_a - \maxwp_b = \alpha$ and $p_a - p_c = \beta_c + \epsilon$ for all $c \neq a,b$. Since $\maxwp_x - \maxwp_b = \alpha \neq \toposet(a,b)$, we know that $x(u) \neq a$. For all alternatives $c \neq a,b$, from Claim~\ref{clm:clm1} and the fact that $\exists u'$ s.t. $\maxwp_{a}' - \maxwp_c' = \beta_c$ and $x(u') = a$, we know that $x(u) = c$ cannot hold. It follows that $x(u) = b$, thus $\exists u$ s.t. $\maxwp_{b} - \maxwp_a = -\alpha $ s.t. $x(u) = b$, thus $-\alpha \in \toposet(b,a)$.
\end{proof}

Intuitively, part 1 means that if for some $u$ s.t. $\maxwp_a - \maxwp_b = \alpha - \epsilon$, $x(u) = a$, it must be the case that for all $u$ s.t. $\maxwp_a - \maxwp_b = \alpha$, $x(u) \neq b$. This is because if $a$ is selected under $\maxwp_a - \maxwp_b = \alpha - \epsilon$, when the difference in the willingness to pay for $a$ vs. $b$ increases by $\epsilon$ thus becomes larger, $a$ still dominates $b$ thus $b$ cannot be selected.

Part 2 means that if for all $u$ s.t. $\maxwp_a - \maxwp_b = \alpha$, $x(u) \neq a$, then there exists some $u$ s.t. $\maxwp_a - \maxwp_b = \alpha$ and $x(u) = b$. This can be proved by constructing a type s.t. the willingness to pay for all alternatives other than $a$ and $b$ in a way that they are all dominated by $a$ so cannot be selected, leaving $b$ to be the only alternative that can be selected.

\begin{claim} \label{clm:clm3}
For every $\alpha$, $\beta$, $\epsilon^{(\alpha)}$, $\epsilon^{(\beta)} \in \setR^n$, $\epsilon^{(\alpha)}$, $\epsilon^{(\beta)} > 0$:
\begin{align*}
	\alpha \hsq - \hsq  \epsilon^{(\alpha)} \hsq  \in  \hsq  \toposet(a,b),~\beta\hsq \hsq  - \epsilon^{(\beta)} \hsq \in \hsq  \toposet(b,c) \hsq  \Rightarrow \hsq  \alpha \hsq  + \hsq  \beta \hsq  - \hsq  \frac{\epsilon^{(\alpha)}\hsq   + \hsq  \epsilon^{(\beta)}}{2}\hsq  \in \toposet(a,c).
\end{align*}
\end{claim}

\begin{proof} 
Choose any $u \in \utilSpace$ s.t. $\maxwp_a - \maxwp_b = \alpha - \epsilon^{(\alpha)}/2$ and $\maxwp_b - \maxwp_c = \beta -\epsilon^{(\beta)}/2$. We know from Claim~\ref{clm:clm1} that $x(u) \neq b$ and $x(u) \neq c$. If the total number of alternatives is more than three, then for all $d \neq a,b,c$, fix some $\delta^{(w)} \in \toposet(a,d)$ and some $\epsilon \in \setR^n$, $\epsilon > 0$, and let $\maxwp_a - \maxwp_d = \delta^{(w)} + \epsilon$. Again by Claim~\ref{clm:clm1}, we know that $x(u) \neq d$. Therefore $x(u) = a$ must hold, and therefore $\alpha + \beta - (\epsilon^{(\alpha)} +\epsilon^{(\beta)} )/2 \in \toposet(a,c)$.
\end{proof}

With the same argument as in the proof for quasi-linear Roberts' theorem, we know that if $\vec{0} \in \toposet(a,b)$ for all pairs $(a,b) \in A\times A$, then the interior of all $\toposet(a,b)$ must be equal. However, $\toposet(a,b)$ does not necessarily include $\vec{0}$ thus needs to be ``shifted" to contain this point. We can similarly define
\begin{align*}
	\gamma(x,y) = \inf \{ q \in \setR~|~q \cdot \vec{1} \in \toposet(a,b) \}.
\end{align*}
It is easy to see that the infimum exists. We first argue that the set is not empty. Given that the social choice rule is onto, we can find some $u \in \utilSpace$ s.t. $x(u) = a$. Let $q = \max_i \{ \maxwp_{i,a} - \maxwp_{i,b} \}$ and find $u'$ s.t. $\maxwp_{i,j}' = \maxwp_{i,j}$ for all $i$ and $j \neq a$, and $\maxwp_{i,a}' = \maxwp_{i,b}+ q$ for all $i$. We know from PAD that $x(u') = a$ must still hold, thus this shows that $q \cdot \vec{1} \in \toposet(a,b)$ thus the set is non-empty. To show that the set of $q$'s is lower-bounded, observe that if this is not the case, the corresponding set $\toposet(b,a)$ would be empty, which contradicts ontoness. 

The rest of the proof in~\citet{lavi2009two} for Roberts' theorem for quasi-linear utility domains analyzes of the sets $\toposet(a,b)$ given the above lemmas and claims. Since these analysis does not depend on the values or the utility functions, the same identical arguments follow through for the parallel domain. This completes the proof of the Robert's theorem on unrestricted parallel domains.

\subsection{Proof of Theorem~\ref{thm:dictatorship}} \label{appx:proof_dictatorship}

\ThmDictatorship*

Before proving the main theorem, we first prove two lemmas. The following first lemma provides a characterization of agent-independent prices for mechanisms that satisfy (P1)-(P5).

\begin{lemma} \label{lem:price_characterization}
Fix any social choice mechanism on an unrestricted parallel domain $\utilSpace$ with choice rule $x(u) \in \arg \max_{a \in \alts} \sum_{i \in N} \affcoeff_i \maxwp_{i,a} + \affcnst_a, ~\forall u \in \utilSpace.$ If the mechanism satisfies (P1)-(P5), the agent-independent prices are determined by: 
\begin{enumerate}[(i)]
	\item for $i\in N$ s.t. $\affcoeff_i > 0$, for any $ u_{-i} \in \utilSpace_{-i}$, let $a_{-i}^\ast \in \arg \max_{a \in \alts} \sum_{j \neq i} \affcoeff_j \maxwp_{j, a} + \affcnst_{a}$, we have $t_{i,a}(u_{-i}) = 1/\affcoeff_i ( \sum_{j \neq i} \affcoeff_{j} \maxwp_{j, a^\ast_{-i}} + \affcnst_{a^\ast_{-i}} - \sum_{j \neq i} \affcoeff_{j} \maxwp_{j, a} - \affcnst_{a})$ for all $a \in \alts$.
	\item for $i \in N$ s.t. $\affcoeff_i = 0$, for any $ u_{-i} \in \utilSpace_{-i}$, s.t. there exists $a_{-i}^\ast\in \alts$ that satisfies $\sum_{j \neq i} \affcoeff_j \maxwp_{j, a} + \affcnst_a < \sum_{j \neq i} \affcoeff_j \maxwp_{j, a^\ast} + \affcnst_{a^\ast}$ for all $ a \neq a_{-i}^\ast$, we have $t_{i,a_{-i}^\ast}(u_{-i}) = 0 \text{ and } t_{i,a}(u_{-i}) = +\infty,~\forall a \neq a_{-i}^\ast$.
\end{enumerate}
\end{lemma}

In other words, for agents s.t. $\affcoeff_i \neq 0$, and for agents s.t. $\affcoeff_i = 0$ but when there is no tie in the choice rule, the agent-independent prices must be determined by \eqref{equ:ai_prices_gwVCG} and \eqref{equ:ai_prices_gwVCG_k0}. As a consequence, when there is no tie, the outcome of the mechanism must be the same as that under the generalized weighted VCG mechanism. 

\begin{proof}

Fix the choice rule as $x(u) \in \arg\max_{a \in \alts} \sum_{i=1}^n \affcoeff_i \maxwp_{i,a} + C_a$, $\forall u \in \utilSpace$ for some non-negative coefficients $\{ \affcoeff_i\}_{i \in N}$ and real constants $\{ \affcnst_a \}_{a \in \alts}$. 

\smallskip 

\noindent{}\textit{Part (i):} Consider any agent $i \in N$ s.t. $\affcoeff_i > 0$.
For all $u_{-i} \in \utilSpace_{-i}$, for any two alternatives $a, b \in \alts$, there exists $u_i \in \utilDomain_i$ such that both $\sum_{j \in N}\affcoeff_j \maxwp_{j,a} + \affcnst_a > \sum_{j \in N} \affcoeff_j \maxwp_{j,c} + \affcnst_c$ and $\sum_{j \in N}\affcoeff_j \maxwp_{j,b} + \affcnst_b > \sum_{j \in N}\affcoeff_j \maxwp_{j,c} + \affcnst_c$ hold, since $\utilDomain_i$ is unrestricted and $k_i > 0$. We know from the choice rule that only alternatives $a$ and $b$ can be selected, and $a$ is selected if
\begin{align}
	& \sum_{j \in N} \affcoeff_j \maxwp_{j,a} + \affcnst_a > \sum_{j \in N}\affcoeff_j \maxwp_{j,b} + \affcnst_b  
	\Leftrightarrow &  \maxwp_{i,a} - \maxwp_{i,b} >  \frac{1}{\affcoeff_i} \left( \sum_{j\neq i} \affcoeff_{j} \maxwp_{j,b} + \affcnst_{b} -  \sum_{j\neq i} \affcoeff_{j} \maxwp_{j,a} - \affcnst_{a} \right). \label{equ:diff_maxwp_1}
\end{align}
From Lemma~\ref{lem:maxwp_as_values}, we also know that agent-maximization requires that $a$ cannot be selected if
\begin{align}
	\maxwp_{i,a} - t_{i,a}(u_{-i}) < \maxwp_{i,b} - t_{i,b}(u_{-i})  \Leftrightarrow \maxwp_{i,a} - \maxwp_{i,b} < t_{i,a}(u_{-i}) - t_{i,b}(u_{-i}). \label{equ:diff_maxwp_2}
\end{align}
%
If the differences in the agent-independent prices satisfies
$
	t_{i,a}(u_{-i}) - t_{i,b}(u_{-i}) > \frac{1}{\affcoeff_i} ( \sum_{j\neq i} \affcoeff_{j} \maxwp_{j,b} + \affcnst_{b} -  \sum_{j\neq i} \affcoeff_{j} \maxwp_{j,a} - \affcnst_{a} ), 
$
there exists $u_i \in \utilDomain_i$ s.t. $\maxwp_{i,a} - \maxwp_{i,b}  \in  ( \frac{1}{\affcoeff_i}( \sum_{j\neq i} \affcoeff_{j} \maxwp_{j,b} + \affcnst_{b} -  \sum_{j\neq i} \affcoeff_{j} \maxwp_{j,a} - \affcnst_{a} ) ,~t_{i,a}(u_{-i}) - t_{i,b}(u_{-i}) )$  in which case both \eqref{equ:diff_maxwp_1} and \eqref{equ:diff_maxwp_2} hold. This is a contradiction. Similarly, we can show that it cannot be the case that 
$
	t_{i,a}(u_{-i}) - t_{i,b}(u_{-i}) < \frac{1}{\affcoeff_i} t( \sum_{j\neq i} \affcoeff_{j} \maxwp_{j,b} + \affcnst_{b} -  \sum_{j\neq i} \affcoeff_{j} \maxwp_{j,a} - \affcnst_{a} ), 
$
thus the price difference must be
\begin{align}
	t_{i,a}(u_{-i}) - t_{i,b}(u_{-i}) = \frac{1}{\affcoeff_i} \left( \sum_{j\neq i} \affcoeff_{j} \maxwp_{j,b} + \affcnst_{b} -  \sum_{j\neq i} \affcoeff_{j} \maxwp_{j,a} - \affcnst_{a}  \right).  \label{equ:price_diff}
\end{align}

Since the choice of $a$ and $b$ are arbitrary, all prices are pinned-down up to a constant, given that the differences between any pair of prices are determined.
Now observe $t_{i,a}(u_{-i})$ is smaller if $\sum_{j\neq i} \affcoeff_{j} \maxwp_{j,a} + \affcnst_{a} $ is larger. From Lemma~\ref{lem:aiprice_characterization} we know that all prices must be non-negative, and that one of the prices must be zero, thus $t_{i,a_{-i}^\ast}(u_{-i})$, the smallest of all, must be exactly zero. Therefore we get:
\begin{align*}
	t_{i,a_{-i}^\ast}(u_{-i}) = 0 = \frac{1}{\affcoeff_i} \left( \sum_{j\neq i} \affcoeff_{j} \maxwp_{j,a^\ast_{-i}} + \affcnst_{a_{-i}^\ast} -  \sum_{j\neq i} \affcoeff_{j} \maxwp_{j,a_{-i}^\ast} - \affcnst_{a_{-i}^\ast}  \right),
\end{align*}
and
\begin{align*}
	t_{i,a}(u_{-i}) = &t_{i,a_{-i}^\ast} + \frac{1}{\affcoeff_i} \left( \sum_{j\neq i} \affcoeff_{j} \maxwp_{j,a^\ast_{-i}} + \affcnst_{a_{-i}^\ast} -  \sum_{j\neq i} \affcoeff_{j} \maxwp_{j,a} - \affcnst_{a}  \right) \\
		=& \frac{1}{\affcoeff_i} \left( \sum_{j\neq i} \affcoeff_{j} \maxwp_{j,a^\ast_{-i}} + \affcnst_{a_{-i}^\ast} -  \sum_{j\neq i} \affcoeff_{j} \maxwp_{j,a} - \affcnst_{a}  \right) .
\end{align*}


\noindent{}\textit{Part (ii):} Now consider $i \in N$ s.t. $\affcoeff_i = 0$. For all $u_{-i} \in \utilSpace_{-i}$ such that $\exists a_{-i}^\ast \in \alts$ that satisfies $\forall a \neq a_{-i}^\ast$, $\sum_{j \neq i} \affcoeff_j \maxwp_{j, a} + \affcnst_a < \sum_{j \neq i} \affcoeff_j \maxwp_{j, a^\ast} + \affcnst_{a^\ast}$,
 we know that no type of agent $i$ in the parallel domain $u_{i} \in \plDomain$ will result in $x(u_i, u_{-i}) = a$ for $a \neq a_{-i}^\ast$. 
Therefore alternative $a$ cannot be the unique agent-maximizing alternative for any $u_i \in \utilDomain_i$, which implies $t_{i,a}(u_{-i}) = +\infty$ must hold. Now we know from Lemma~\ref{lem:aiprice_characterization} that $t_{i,a_{-i}^\ast}(u_{-i}) = 0$ must be true, since one of the prices must be zero. 
\end{proof}

Given a utility domain $\utilSpace = \prod_{i=1}^n \utilDomain_i$, denote the parallel sub-domain of each agent $i \in N$ as $\udomainintpl_i \triangleq \utilDomain_i \cap \plDomain$, and let $\uspaceintpl \triangleq \utilSpace \cap \plSpace$ be the subspace of $\utilSpace$ containing all parallel type profiles. 
If the utility domain $\utilSpace$ satisfies (C1), i.e. when $\utilDomain_i$ contains an unrestricted parallel domain, then each of the $\udomainintpl_i $ is an unrestricted parallel domain, and $\uspaceintpl$ is also unrestricted.
For any social choice mechanism $(x,t)$ on $\utilSpace$, we show that its restriction on $\uspaceintpl$ must inherit its good properties.

\begin{lemma} \label{lem:subdomain_onto} Fix any social choice mechanism $(x,t)$ under (P1)-(P3) on a utility domain $\utilSpace$ that satisfies (C1). The restriction of $(x,t)$ on the parallel subdomain $\uspaceintpl = \utilSpace \cap \plSpace$ also satisfies (P1)-(P3). Moreover, if $(x,t)$ also satisfies (P4) and (P5), then (P4) and (P5) are also satisfied by its restriction on $\uspaceintpl$. 
\end{lemma}

\begin{proof} For any mechanism that satisfies (P1)-(P3), it is immediate that (P1) DSIC (P2) deterministic must also hold for its restriction on any subdomain. Similarly, for a mechanism with (P1)-(P5), its restriction on any subdomain must also satisfy (P4) IR and (P5) no subsidy. What is left to show is that any mechanism that with (P1)-(P3) on a domain $\utilSpace$ that satisfies condition (C1), its restriction on $\uspaceintpl$ must also be onto. 
Assume toward a contradiction, that there exists $a^\ast \in \alts$ s.t. $\forall u \in \uspaceintpl$, $x(u) \neq a^\ast$. 
We prove by induction that the following statement holds for all $\ell \leq n-1$:
\begin{itemize}
	\item[$\FF_\ell$:] $\forall i \in N$, $\forall u_{-i} \in \utilSpace_{-i}$ such that $|\{ i' \in N ~|~ i' \neq i,~u_{i'} \notin \udomainintpl_i \}| \leq \ell$, we have $t_{i,a^\ast}(u_{-i}) = \infty$.
\end{itemize}
This implies that $a^\ast$ cannot be the agent-maximizing alternative for any agent under any type profile $u \in \utilSpace$, thus $a^\ast$ cannot be selected, and this violates the ontoness of $(x,t)$ on $\utilSpace$.

We first prove $\FF_0$. When $| \{ i' \in N ~|~ i' \neq i,~u_{i'} \notin \udomainintpl_i \}| \leq 0$, we know $u_{-i} \in \uspaceintpl_{-i}$.
For any $i \in N$ and any $u_{-i} \in \uspaceintpl_{-i}$, assuming $t_{i,a^\ast}(u_{-i}) < \infty$, there exists $u_i \in \udomainintpl_i$ s.t. $\maxwp_{i,a^\ast} - t_{i,a^\ast}(u_{-i}) > \maxwp_{i,a} - t_{i,a}(u_{-i})$ for all $a \neq a^\ast$, given that $\udomainintpl_i$ is unrestricted. In this case, given the parallel profile $(u_i, u_{-i})$, $a^\ast$ is the unique agent-maximizing alternative for agent $i$ thus has to be selected. This contradicts the assumption that $x(u) \neq a^\ast$ for all $u \in \uspaceintpl$, thus $t_{i,a^\ast}(u_{-i}) = \infty$ must hold. 

Now assume $\FF_{\ell-1}$ holds for some $\ell$ s.t. $1 \leq \ell \leq  n-2$, we show that $\FF_{\ell}$ also holds. W.l.o.g., we consider agent $1$ and some $u_{-1}$ s.t. $u_i \in \udomainintpl_i$ for all $i \geq \ell + 2$. In this case, only $u_2, \dots, u_{\ell+1}$ can be non-parallel. Now consider agent $2$, and any $u_1 \in \udomainintpl_i$, we know in $u_{-2} = (u_1, u_{-1,-2})$, $|\{ i \in N | i \neq 2, u_i \notin \udomainintpl_i \}| \leq \ell - 1$. As a result, $\FF_{\ell-1}$ implies that $t_{2, a^\ast} (u_{-2}) = \infty$, thus alternative $a^\ast$ cannot be agent-maximizing for agent $2$ and therefore cannot be selected in the economy $(u_1, u_{-1})$. Since this holds for any $u_1 \in \udomainintpl_1$, with the same arguments that we proved $\FF_0$, we conclude $t_{1,a^\ast}(u_{-1}) = \infty$ must hold as well, since otherwise $a^\ast$ would be the unique agent-maximizing alternative for some parallel type $u_1$ and this violates DSIC. This proves $\FF_{\ell-1} \Rightarrow \FF_\ell$, and therefore completes the proof of this lemma.
\end{proof}



We are now ready to prove the main theorem.

\begin{proof} [Proof of Theorem~\ref{thm:dictatorship}]
Recall that $\udomainintpl_i \triangleq \utilDomain_i \cap \plDomain$ for each $i \in N$, and $\uspaceintpl \triangleq \utilSpace \cap \plSpace$. For any social choice mechanism $(x,t)$ on $\utilSpace$ that satisfies (P1)-(P5), Lemma~\ref{lem:subdomain_onto} implies that its restriction on $\uspaceintpl$ must also satisfy (P1)-(P5). Theorem~\ref{thm:roberts} then guarantees that there exist non-negative coefficients $\{ \affcoeff_i \}_{i \in N}$, not all of them zero, and real constants $\{ \affcnst_a \}_{a \in \alts}$ s.t. $x(u) \in \arg\max_{a \in \alts} \sum_{i\in N} \affcoeff_i \maxwp_{i,a} + \affcnst_a$ for all $u \in \uspaceintpl$. 
Although this does not immediately determine the outcome of the mechanism for any non-parallel type profile $u \in \utilSpace \backslash \uspaceintpl$, we do know from agent independence that for any agent $i$, and any parallel profile for the other agents $u_{-i} \in \uspaceintpl_{-i}$, the agent-independent prices agent $i$ faces $\{t_{i,a}(u_{-i}) \}_{a \in \alts}$ must be characterized as in Lemma~\ref{lem:price_characterization}. 
We use this characterization and condition (C2) $\utilDomain_i \not\subset \plDomain$ for at least $n-1$ agents to prove the dictatorship result in the following two steps:
\begin{itemize}
	\item Step 1: the number of agents s.t. $\affcoeff_i \neq 0$ is  exactly one.
	\item Step 2: the only agent with $\affcoeff_i \neq 0$ must be a fixed price dictator.
\end{itemize}

\medskip

\noindent{}\textit{Step 1:} 
We know from Theorem~\ref{thm:roberts} that there exists at least one agent with $\affcoeff_i > 0$. 
Assume towards a contradiction, that there exist at least two agents, which we name agent $1$ and agent $2$, for whom $\affcoeff_1,\:\affcoeff_2 > 0$. For at least one of them, say agent 2, $\utilDomain_2 \subseteq \plDomain$ does not hold due to condition (C2), thus there exists a non-parallel type $u_2^\ast \in \utilDomain_2 \backslash \udomainintpl_2$. 
We prove that for some parallel profile for the rest of the agents $u_{-1,-2}^\ast \in \uspaceintpl_{-1,-2}$, there do not exist agent-independent prices $\{ t_{1,a}(u_2^\ast,u_{-1,-2}^\ast) \}_{a \in \alts}$ for agent $1$, such that for any $u_1 \in \udomainintpl_1$, there exists an alternative that is agent-maximizing for all agents in economy $(u_1, u_2^\ast, u_{-1,-2}^\ast)$. This contradicts DSIC, therefore there is exactly one agent s.t. $\affcoeff_i > 0$.

We assume w.l.o.g. that  alternative $a$ is one of agent 2's favorite alternatives at zero payment: $a \in \arg\max_{a'\in \alts} v_{2,a'}^\ast$ (where $v^\ast_{2,a} \triangleq u_{2,a}^\ast(0)$ for all $a \in \alts$). For any parallel type $u_2 \in \plDomain$ s.t. $a \in \arg \max_{a'\in \alts} \{ v_{2,a'} \}$, we know from the definition of the parallel domain that for any alternative $a' \in \alts$, $u_{2,a'}(z) =  u_{2,a}(z + \maxwp_{2,a} - \maxwp_{2,a'}) = u_{2,a}(z + u_{2,a}^{-1}(v_{2,a'}))$ must hold for all $z \leq \maxwp_{2,a'}$. 
$u_2^\ast \notin \plDomain$ implies that there exists some alternative  $b \notin \arg\min_{a' \in \alts} v_{2,a'}^\ast$ and some price $\payment^\ast \in (0, \maxwp_{2,b}^\ast]$ s.t. $u_{2,b}^\ast(\payment^\ast) \neq u_{2,a}^\ast(\payment^\ast + (u_{2,a}^\ast)^{-1}(v_{2,b}^\ast))$. 
Let $w \triangleq u_{2,b}^\ast(\payment^\ast)$ and define $\Delta_1 \triangleq (u_{2,a}^\ast)^{-1}(v_{2,b}^\ast)$ and $\Delta_2 \triangleq (u_{2,a}^\ast)^{-1}(w) - \payment^\ast$. We assume $u_{2,b}^\ast(\payment^\ast) < u_{2,a}^\ast(\payment^\ast + \Delta_1)$, in which case $\Delta_2 > \Delta_1 \geq 0$, as illustrated in Figure~\ref{fig:thmDictatorProof_Step4}. The other direction can be proved in the same way.

\begin{figure}[t!]
\vspace{-0.0em}
\centering     

\begin{tikzpicture}[scale = 1][font = \normalsize]
\draw[ ->] (-0.1,0) -- (8.5,0) node[anchor=north] {$\payment$};

\draw[ ->] (0,-0.4) -- (0, 3.7) node[anchor=west] {$u_{2,a'}^\ast(\payment)$};

\draw  [line width=0.2mm, -]  (0, 3.2) to[out=-45, in=-200] (5, 0.8) to[out=-20, in=-190] (7.5, 0.2);

\draw  [line width=0.2mm, dashed]  (0, 2.2) to[out=-55, in=-200] (3.5, -0.2); 

\draw[line width=0.2mm, dashdotted](0, 0.3) parabola[bend at end] (2, -0.2);

\draw[dotted](0, 2.2) -- (1.8, 2.2) ;
\draw[dotted](1.35, 2.4) -- (1.35, 1.8) ;

\draw[decoration={brace,mirror,raise=5pt},decorate]
  (0,2.4) -- node[below=5pt] { { \footnotesize $\Delta_1$}} (1.35,2.4);
  
\draw[dotted](-0.1, 0.9) -- (5.2, 0.9) ;
\draw[dotted](1.2, 1.2) -- (1.2, -0.15) ;
\draw (1.2, -0.1) node[anchor=north] { {  $\payment^\ast$ } };
\draw (0, 0.9) node[anchor=east] { {  $w$ } };
  
\draw[decoration={brace,mirror,raise=5pt},decorate]
  (1.25,1.1) -- node[below=5pt] { { \footnotesize $\Delta_2$}} (4.7,1.1);

\draw[dotted](-0.1, 0.3) -- (7.4, 0.3) ;
\draw[dotted](2.25, 0.5) -- (2.25, -0.25) ;
\draw[dotted](6.85, 0.5) -- (6.85, -0.25) ;
\draw (2.25, -0.18) node[anchor=north] { {  $\maxwp_{2,b}^\ast$ } };
\draw (6.85, -0.18) node[anchor=north] { {  $\maxwp_{2,a}^\ast$ } };

\draw[dotted](2.7, 1.9) -- (2.7, 1.1) ;
\draw (3, 1.8) node[anchor=south] { {$\payment^\ast + \Delta_1$ } };
\draw[dotted](-0.1, 1.55) -- (2.9, 1.55) ;
\draw (0, 1.5)node[anchor = east] {{ \small $u_{2,a}^\ast(\payment^\ast + \Delta_1)$}};

\draw (0, 3.2)node[anchor = east] {{  $v_{2,a}^\ast$}};
\draw (0, 2.2)node[anchor = east] {{  $v_{2,b}^\ast$}};
\draw (0, 0.25)node[anchor = east] {{  $v_{2, \worstalt_2}^\ast$ }};

\draw[line width=0.2mm, -] (6.2, 3.5) -- (6.85, 3.5) node[anchor=west] {$u_{2,a}^\ast(\payment)$};
\draw[line width=0.2mm, dashed] (6.2, 2.75) -- (6.85, 2.75) node[anchor=west] {$u^\ast_{2,b}(\payment)$};
\draw[line width=0.2mm, dashdotted] (6.2, 2) -- (6.85,2) node[anchor=west] {$u^\ast_{2, \worstalt_2}(\payment)$};

\end{tikzpicture}
\caption{Illustration of $u_2^\ast \in \utilDomain_2 \backslash \udomainintpl_2$, for Step~1 of the proof of Theorem~\ref{thm:dictatorship}.
\label{fig:thmDictatorProof_Step4}} 
\end{figure}

Fixing the type of the agents other than 1 and 2 to be parallel with zero willingness to pay on all alternatives, i.e. for all $i \geq 3$, let $u_i^\ast \in \udomainintpl_i$  be such that $\maxwp_{i,a'}^\ast =0$ for all $a' \in \alts$. Such types exist since $\udomainintpl_i$ is unrestricted for each $i \in N$. Denote $u^\ast_{-1} = (u_2^\ast, u_3^\ast, \dots, u_n^\ast)$, and let $\epsilon$ be some small positive number s.t. $ 0 < \epsilon < (\Delta_2 - \Delta_1)/2$. 
We prove:
\begin{itemize}
	\item Step 1.1: $t_{1,b}(u_{-1}^\ast) - t_{1,a}(u_{-1}^\ast) \leq \affcoeff_2/\affcoeff_1 (\Delta_1 + \epsilon) + (\affcnst_a - \affcnst_b)/ \affcoeff_1 $,
	\item Step 1.2: $t_{1,b}(u_{-1}^\ast) - t_{1,a}(u_{-1}^\ast) \geq \affcoeff_2/\affcoeff_1 (\Delta_2 - \epsilon) + (\affcnst_a - \affcnst_b)/ \affcoeff_1 $.
\end{itemize}

Since $\affcoeff_2/\affcoeff_1 (\Delta_2 - \epsilon) + (\affcnst_a - \affcnst_b)/ \affcoeff_1 > \affcoeff_2/\affcoeff_1 (\Delta_1 + \epsilon) + (\affcnst_a - \affcnst_b)/ \affcoeff_1$, we know that this is a contradiction, thus the number of agents for whom $\affcoeff_i > 0$ cannot be more than one. 

\medskip

\noindent{}\textit{Step 1.1:} Assume towards a contradiction that $t_{1,b}(u_{-1}^\ast) - t_{1,a}(u_{-1}^\ast) > \affcoeff_2/\affcoeff_1 (\Delta_1 + \epsilon) + (\affcnst_a - \affcnst_b)/ \affcoeff_1 $ and consider a parallel type $u_1 \in \udomainintpl_1$ of agent 1 with the following willingness to pay:
\begin{align}
	\maxwp_{1,c} &= 0,~\forall c \neq a,b, \label{equ:p_1c_thm_dic_proof} \\
	\maxwp_{1,b} &= \max_{c \neq a,b} \left\lbrace \frac{\affcoeff_2}{\affcoeff_1} \maxwp_{2,c}^\ast + \frac{\affcnst_c - \affcnst_b}{\affcoeff_1}  \right\rbrace + \delta, \label{equ:p_1b_thm_dic_proof} \\
	\maxwp_{1,a} &= \maxwp_{1,b} - \left(\frac{\affcoeff_2}{\affcoeff_1} (\Delta_1 + \epsilon/2) + \frac{ \affcnst_a - \affcnst_b }{\affcoeff_1}  \right),  \label{equ:p_1a_thm_dic_proof}
\end{align}
where $\delta$ is strictly positive, and large enough s.t. $\maxwp_{1,a}$ and $\maxwp_{1,b} $ as defined are both non-negative. Such $u_1$ is guaranteed to exist since $\udomainintpl_1$ is unrestricted. We know from \eqref{equ:p_1a_thm_dic_proof} and the assumption $t_{1,b}(u_{-1}^\ast) - t_{1,a}(u_{-1}^\ast) > \affcoeff_2/\affcoeff_1 (\Delta_1 + \epsilon) + (\affcnst_a - \affcnst_b)/ \affcoeff_1 $ that:
\begin{align*}
	 & \maxwp_{1,a} - t_{1,a}(u_{-1}^\ast) - \left(\maxwp_{1,b} - t_{1,b}(u_{-1}^\ast) \right) \\
	= & \maxwp_{1,a} - \maxwp_{1,b} +  \left( t_{1,b}(u_{-1}^\ast) - t_{1,a}(u_{-1}^\ast)\right)  \\
	> & - \left(\frac{\affcoeff_2}{\affcoeff_1} (\Delta_1 + \epsilon/2) + \frac{ \affcnst_a - \affcnst_b }{\affcoeff_1}  \right)  +  \left(\frac{\affcoeff_2}{\affcoeff_1} (\Delta_1 + \epsilon) + \frac{ \affcnst_a - \affcnst_b }{\affcoeff_1}  \right) \\
	 = &\frac{\affcoeff_2}{2\affcoeff_1}\epsilon > 0,
\end{align*}
thus $\maxwp_{1,a} - t_{1,a}(u_{-1}^\ast) > \maxwp_{1,b} - t_{1,b}(u_{-1}^\ast)$. We conclude according to Lemma~\ref{lem:maxwp_as_values} that with prices $ \{ t_{1,a'}(u_{-1}^\ast) \}_{a'\in \alts}$, $b$ cannot be an agent-maximizing alternative for agent $1$.

We now prove that $b$ is the only agent-maximizing alternative for agent 2, therefore no alternative can be agent-maximizing for both agents, leading to a contradiction to DSIC. 
First, by assumption, $\maxwp_{i,a'}^\ast = 0$ for all $i \neq 1,2$ and all $a' \in \alts$, thus
%
%
when the type profile of the rest of the economy is given by $(u_1, u_3^\ast,\dots, u_n^\ast)$, we have
\begin{align*}
	\arg \max_{a'\in \alts} \left\lbrace \affcoeff_1 \maxwp_{1,a'} + \sum_{i \geq 3} \affcoeff_i \maxwp_{i,a'}^\ast + \affcnst_{a'} \right\rbrace = \arg \max_{a' \in \alts} \left\lbrace \affcoeff_1 \maxwp_{1,a'} + \affcnst_{a'} \right\rbrace.
\end{align*}
We can now check that $b 
\in \arg\max_{a' \in \alts} \left\lbrace \affcoeff_1 \maxwp_{1,a'} + \affcnst_{a'} \right\rbrace $. From \eqref{equ:p_1a_thm_dic_proof} we know
\begin{align*}
	&\affcoeff_1 \maxwp_{1,b} + \affcnst_b - ( \affcoeff_1 \maxwp_{1,a} + \affcnst_{a} )  = \affcoeff_1(\maxwp_{1,b} - \maxwp_{1,a})  + \affcnst_b - \affcnst_a \\
	=& \affcoeff_2(\Delta_1 + \epsilon/2) + (\affcnst_a - \affcnst_b) - (\affcnst_a - \affcnst_b)  = \affcoeff_2(\Delta_1 + \epsilon/2) > 0,
\end{align*}
thus $\affcoeff_1 \maxwp_{1,b} +\affcnst_b  > \affcoeff_1 \maxwp_{1,a} + \affcnst_{a}$. Moreover, for any $c \neq a,b$, we know from \eqref{equ:p_1c_thm_dic_proof} and \eqref{equ:p_1b_thm_dic_proof} that
\begin{align*}
	&\affcoeff_1 \maxwp_{1,b} + \affcnst_b - ( \affcoeff_1 \maxwp_{1,c} + \affcnst_{c} )  = \affcoeff_1(\maxwp_{1,b} - \maxwp_{1,c})  + \affcnst_b - \affcnst_c \\
	> & \affcoeff_1 \left( \frac{\affcoeff_2}{\affcoeff_1} \maxwp_{2,c}^\ast + \frac{\affcnst_c - \affcnst_b}{\affcoeff_1}  \right) + \affcnst_b - \affcnst_c = \affcoeff_2 \maxwp_{2,c}^\ast \geq 0,
\end{align*}
therefore, $\affcoeff_1 \maxwp_{1,b} +\affcnst_b > \affcoeff_1 \maxwp_{1,c} + \affcnst_{c}$ holds for all $c \neq a,b$. Now we know $b \in \arg\max_{a' \in \alts} \left\lbrace \affcoeff_1 \maxwp_{1,a'} + \affcnst_{a'} \right\rbrace $ which implies $t_{2,b}(u_1, u_{-1,-2}^\ast) = 0$, according to Lemma~\ref{lem:price_characterization}. Thus we know the utility agent 2 gets from alternative $b$ at the current price is: $u_{2,b}^\ast(t_{2,b}(u_1, u_{-1,-2}^\ast))=u_{2,b}^\ast(0)  = v_{2,b}^\ast$. 


For alternative $a$, we know from \eqref{equ:price_diff} and \eqref{equ:p_1a_thm_dic_proof} that
\begin{align*}
	& t_{2,a}(u_1, u_{-1,-2}^\ast)  =  t_{2,b}(u_1, u_{-1,-2}^\ast) + \frac{1}{\affcoeff_2} \left( \affcoeff_1(\maxwp_{1,b} - \maxwp_{1,a}) + \affcnst_b - \affcnst_a \right) \\
	=& 0 +  \frac{1}{\affcoeff_2} \left(\affcoeff_2 (\Delta_1 + \epsilon/2) +\affcnst_a - \affcnst_b  +  \affcnst_b - \affcnst_a  \right) = \Delta_1 + \epsilon/2 > \Delta_1.
\end{align*}
Therefore, $u_{2,a}^\ast(t_{2,a}(u_1, u_{-1,-2}^\ast)) < u_{2,a}^\ast(\Delta_1) = v_{2,b}^\ast$. For all other alternatives $c \neq a,b$, we know from \eqref{equ:price_diff}, \eqref{equ:p_1c_thm_dic_proof}  and \eqref{equ:p_1b_thm_dic_proof} that
\begin{align*}
	& t_{2,c}(u_1, u_{-1,-2}^\ast) =  t_{2,b}(u_1, u_{-1,-2}^\ast) + \frac{1}{\affcoeff_2} \left( \affcoeff_1(\maxwp_{1,b} - \maxwp_{1,c}) + \affcnst_b - \affcnst_c\right) \\
	> & 0 + \frac{1}{\affcoeff_2} \left( \affcoeff_1 \left( \frac{\affcoeff_2}{\affcoeff_1} \maxwp_{2,c}^\ast + \frac{\affcnst_c - \affcnst_b}{\affcoeff_1} \right) + \affcnst_b - \affcnst_c\right) = \maxwp_{2,c}^\ast.
\end{align*}
Therefore, $u_{2,c}^\ast(t_{2,c}(u_1, u_{-1,-2}^\ast)) < \min_{a' \in \alts} v_{2,a'}^\ast \leq v_{2,b}^\ast$ for all $c \neq a,b$. This proves that $\{b\} = \arg\max_{a'\in \alts} u_{2,a'}^\ast(t_{2,a'}(u_1, u_{-1,-2}^\ast)) $, thus completes the proof of part Step 1.1.

\smallskip

\noindent{}\textit{Step 1.2:} Assume for contradiction that $t_{1,b}(u_{-1}^\ast) - t_{1,a}(u_{-1}^\ast) < \affcoeff_2/\affcoeff_1 (\Delta_2 - \epsilon) + (\affcnst_a - \affcnst_b)/ \affcoeff_1 $. As discussed above, $b$ cannot be the least preferred alternative at zero price according to $u_2^\ast$, thus we assume w.l.o.g. $m \in \arg\min_{a' \in \alts} v_{2,a'}^\ast$. Consider the type $u_1 \in \udomainintpl_1$ s.t. 
\begin{align*}
	& \maxwp_{1,c} = 0,~\forall c \neq a,b,m, \\
	& \maxwp_{1,a} =  \max_{c \neq a,b,m} \left\lbrace \frac{\affcoeff_2}{\affcoeff_1} \maxwp_{2,c}^\ast + \frac{\affcnst_c - \affcnst_a}{\affcoeff_1}  \right\rbrace + \delta, \\
	& \maxwp_{1,b} = \maxwp_{1,a} + \left(\frac{\affcoeff_2}{\affcoeff_1} (\Delta_2 - \epsilon/2) + \frac{ \affcnst_a - \affcnst_b }{\affcoeff_1}  \right),  \\
	& \maxwp_{1,m} = \maxwp_{1,b} + \frac{1}{\affcoeff_1}(\affcoeff_2 z^\ast + \affcnst_b - \affcnst_m),
\end{align*} 
where $\delta $ is some non-negative number such that $\min\{ \maxwp_{1,a},\maxwp_{1,b},\maxwp_{1,m} \} \geq 0$.\footnote{If the number of alternatives is exactly 3, then we may set $\maxwp_{1,a} = \delta$ where $\delta \in \setR$ guarantees $\min\{ \maxwp_{1,a},\maxwp_{1,b},\maxwp_{1,m} \} = 0$, so that the smallest willingness to pay among all alternatives is zero. The rest of the proof remains the same.} Similar to the proof of part Step 1.1, for $u_1$ as constructed, we can show that $a$ cannot be the agent-maximizing alternative for agent 1 given prices $\{t_{1,a'}(u_{-1}^\ast)\}_{a'\in \alts}$ since $\maxwp_{1,b} - t_{1,b}(u_{-1}^\ast) > \maxwp_{1,a} - t_{1,a}(u_{-1}^\ast) $ thus $a$ is not the maximizer of $\maxwp_{1,a'} - t_{1,a'}(u_{-1}^\ast)$.
%
%
Moreover, we can prove that $a$ \emph{is} the unique agent-maximizing alternative for agent 2, by showing:
\begin{enumerate}[1)]
	\item  $m \in \arg\max_{a' \in \alts} \{ \affcoeff_1 \maxwp_{1,a'} + \affcnst_{a'} \}$ which implies $t_{2,m}(u_1, u_{-1,-2}^\ast) = 0$, $t_{2,a'}(u_1, u_{-1,-2}^\ast) = 1/\affcoeff_2(\affcoeff_1 \maxwp_{1,m} + \affcnst_m - \affcoeff_1 \maxwp_{1,a'} - \affcnst_{a'})$ for all $a' \in \alts$, and  $u_{2,m}^\ast(t_{2,m}(u_1, u_{-1,-2}^\ast)) = \min_{a' \in \alts}v_{2,a'}^\ast \leq w$.
	\item $t_{2,b}(u_1, u_{-1,-2}^\ast) = \payment^\ast$ thus $u_{2,b}^\ast(t_{2,b}(u_1, u_{-1,-2}^\ast)) = u_{2,b}^\ast(\payment^\ast) = w$,
	\item $t_{2,c}(u_1, u_{-1,-2}^\ast) > \maxwp_{2,c}^\ast$, thus $u_{2,c}^\ast(t_{2,c}(u_1, u_{-1,-2}^\ast)) < \min_{a' \in \alts}v_{2,a'}^\ast \leq w$ for all $ c \neq a,b,m$, and
	\item $t_{2,a}(u_1, u_{-1,-2}^\ast) = \Delta_2 + \payment^\ast - \epsilon/2 < \Delta_2+ \payment^\ast$,  therefore $u_{2,a}^\ast(t_{2,a}(u_1, u_{-1,-2}^\ast))>  u_{2,a}^\ast(\Delta_2+ \payment^\ast) = w \geq \max_{a' \neq a } \{ u_{2,a'}^\ast(t_{2,a'}(u_1, u_{-1,-2}^\ast)) \}$.
\end{enumerate}
This shows that no alternative is agent-maximizing for both agents 1 and 2, and completes the proof of this Step 1.2, and also Step 1.

\bigskip

\noindent{}\textit{Step 2:} The mechanism must be a fixed price dictatorship.  

So far, we have proved  that for any mechanism  $(x,t)$ satisfying (P1)-(P5), its restriction on the parallel subdomain $\uspaceintpl$  must be an affine maximizer of willingness to pay with coefficients $\{\affcoeff_i\}_{i \in N} $ and constants $\{\affcnst_a\}_{a \in \alts}$, where $\affcoeff_i > 0$ for exactly one agent. Let's name her agent 1 and let $\vec{\payment}$ be a vector of fixed prices in $ \setR_{\geq 0}^m$ s.t. 
\begin{align}
	z_a \triangleq \frac{1}{\affcoeff_1} \left( \max_{a' \in \alts} \{ \affcnst_{a'} \} - \affcnst_a \right), ~\forall a \in \alts. \label{equ:fixed_prices}
\end{align}
To show that agent $1$ is a fixed price dictator, i.e. $x(u) \in \arg \max_{a \in \alts} u_{1,a}(\payment_a)$ and $ t_{1}(u) =\payment_{x(u)}$ for all $u \in \utilSpace$, agent-maximization implies that it is sufficient to show for all $u_{-1} \in \utilSpace_{-1}$, $t_{1,a}(u_{-1}) = \payment_{a}$ holds. We prove this by induction on the number of agents whose types are not parallel in the profile $u_{-1}$.
For any $\ell = 0, 1, \dots, n-1$, let the induction statements be 
\begin{enumerate}
	\item [$\GG_\ell$:] For all $u_{-1} \in \utilSpace_{-1}$ s.t. $| \{ i \in N |  i \neq 1, u_i \notin \plDomain \}|\leq \ell$, $t_{1,a}(u_{-1}) = \payment_{a}$ holds for all $ a \in \alts$.
	\item [$\HH_\ell$:] For all $ i \neq 1$, for all $ u_{-i} \in \utilSpace_{-i}$ such that (I) $| \{ j \in N | j \neq i,~u_j \notin \plDomain \}| \leq \ell$, and (II) $\exists a^\ast$ s.t. $u_{1,a^\ast}(\payment_{a^\ast}) > u_{1,a}(\payment_a)$ for all $a \neq a^\ast$, we have $t_{i,a^\ast}(u_{-i}) = 0$, and $t_{i,a}(u_{-i}) = +\infty$ for all $a \neq a^\ast$.
\end{enumerate}

We first observe that Lemma~\ref{lem:price_characterization} implies $\GG_0$. For $u_{-1}$ s.t. $| \{ i \in N |  i \neq 1, u_i \notin \plDomain \}| = 0$, $u_{-1} \in \uspaceintpl_{-1}$, thus part (i) of Lemma~\ref{lem:price_characterization}, agent-independence and the fact $\affcoeff_i = 0$ for $i \neq 1$ imply that $\forall a \in A$:
%
\begin{align*}
	t_{1,a}(u_{-1}) =& \frac{1}{\affcoeff_1} \max_{a' \in \alts} \left\lbrace \sum_{i \neq 1} \affcoeff_i \maxwp_{i,a'} + \affcnst_{a'} \right\rbrace  - \frac{1}{\affcoeff_1}  \left( \sum_{i \neq 1} \affcoeff_i \maxwp_{i,a} + \affcnst_{a} \right) = \frac{1}{\affcoeff_1} \left(  \max_{a' \in \alts} \affcnst_{a'} - \affcnst_a \right) = \payment_a.
\end{align*}
%
%
We show in the following two steps that  $\GG_{\ell} \Rightarrow \HH_{\ell}$ and  $\HH_{\ell-1} \Rightarrow \GG_{\ell}$. This implies that $\GG_\ell$ holds for $\ell = n-1$, and completes the proof of the theorem. $\HH_{n-1}$ also implies that when the dictator has a unique most preferred alternative at the fixed prices $\{ \payment_a \}_{a \in \alts}$, the rest of the agents cannot be charged any payment. 

\medskip

\noindent{}\emph{Step 2.1: $\GG_{\ell} \Rightarrow \HH_{\ell}$ for $0 \leq \ell \leq n-1$.}

W.l.o.g., consider agent $i=2$, and some $u_{-2} \in \utilSpace_{-2}$ s.t. $|\{ j \in N | j \neq 2, u_j \notin \plDomain \}| = \ell$. For any $u_2 \in \udomainintpl_2$, in the economy $(u_2, u_{-2})$, the number of agents in the profile $u_{-1}$ with types outside of the parallel domain is at most $\ell$. Since $|\{ j \in N | j \neq 2, u_j \notin \plDomain \}| = \ell$ and $u_2$ is parallel, $|\{ j \in N | j \neq 1, u_j \notin \plDomain \}| = \ell-1$ if $u_1 \notin \plDomain$, and $|\{j \in N | j \neq 1, u_j \notin \plDomain \}| = \ell$ if $u_1 \in \plDomain$. We know from $\GG_\ell$ that $t_{i,a}(u_{-1}) = \payment_a$ for all $a \in \alts$.

If there exists a unique agent-maximizing alternative for agent $1$ given fixed prices $\vec{z}$, i.e. $\exists a^\ast$ s.t. $  u_{1,a^\ast}(\payment_{a^\ast}) > u_{1,a}(\payment_a)$, $\forall a \neq a^\ast$, the only alternative that can be selected in the economy $(u_2, u_{-2})$ is $a^\ast$. This implies that for all $u_2 \in \udomainintpl_2$, $x(u_2, u_{-2}) = a^\ast$ must hold. 
Since $\udomainintpl_2$ is unrestricted, in order for $a^\ast$ to be agent-maximizing for agent $2$ for any $u_2 \in \udomainintpl_2$, we must have $t_{2,a}(u_{-2}) = \infty$ for all $a \neq a^\ast$. Since prices must be standard, $t_{2,a^\ast}(u_{-2}) = 0$ and this completes the proof of $\GG_{\ell} \Rightarrow \HH_{\ell}$. 

\smallskip
\noindent{}\emph{Step 2.2: $\HH_{\ell-1} \Rightarrow \GG_{\ell}$ for all $1 \leq \ell \leq n-1$.}

Let there be $\ell$ entries in $u_{-1}$ that are outside of the parallel domain, and w.l.o.g. assume that $u_2, \dots, u_{\ell+1} \notin \plDomain$. Assume that $\exists a \in \alts$ s.t. $t_{1,a}(u_{-1}) \neq \payment_a$, we first show a contradiction for the case that $t_{1,a}(u_{-1}) > \payment_a$, and then show that the other direction cannot hold either.
First, note that it cannot be the case if $t_{1,a'}(u_{-1}) > \payment_{a'}$ for all $a' \in \alts$. This is because $\vec{z}$ as defined in \eqref{equ:fixed_prices} is a vector of standard prices with the minimal entry equal to $0$. If $t_{1,a'}(u_{-1}) > \payment_{a'}$ for all $a' \in \alts$, we know $t_{1,a'}(u_{-1}) > 0$ for all $a'$, and this violates Lemma~\ref{lem:aiprice_characterization}. W.l.o.g., we assume $t_{1,b}(u_{-1}) \leq \payment_{b}$.

Denote $\epsilon \triangleq t_{1,a}(u_{-1}) - \payment_a$, we know $\epsilon > 0$. Consider some parallel type of agent 1, $u_1 \in \udomainintpl_1$, where the willingness to pay is of the form:
\begin{align*}
	\maxwp_{1,a} =&  \payment_a + \epsilon / 2, \\
	\maxwp_{1,b} =&  \payment_b + \epsilon / 3,  \\
	\maxwp_{1,c} =&  0, ~\forall c \neq a,b.
\end{align*}
This gives us $\maxwp_{1,b} - \payment_b = \epsilon / 3$, $\maxwp_{1,c} - \payment_c = - \payment_c , ~\forall c \neq a,b$, and therefore $\maxwp_{1,a} - \payment_a = \epsilon / 2 > \maxwp_{1,a'} - \payment_{a'} $  for all $a' \neq a$. Given Lemma~\ref{lem:maxwp_as_values}, we know that $a$ is the unique agent-maximizing alternative for agent $1$ under the vector of prices $\vec{\payment}$.

In the economy $(u_1, u_2, \dots, u_{\ell+1}, u_{\ell+2}, \dots, u_n)$, there are $\ell-1$ entries in the profile $u_{-2}$ that are outside of the parallel domain: $|\{i \in N | i \neq 2,~u_i \notin \plDomain \}| = \ell-1$.
$\HH_{\ell-1}$ implies that $t_{2,a}(u_{-2}) = 0$, $t_{2,a'}(u_{-2}) = +\infty$ for all $a' \neq a$ thus $a$ is the unique agent-maximizing alternative for agent 2. However, $a$ cannot be agent-maximizing for agent $1$, since 
\begin{align*}
	& \maxwp_{1,a} - t_{1,a}(u_{-1}) - (\maxwp_{1,b} - t_{1,b}(u_{-1})) 
	 \leq   \maxwp_{1,a} - (\payment_a + \epsilon) - \maxwp_{1,b} + \payment_b \\
	 = & \payment_a + \epsilon/2 - (\payment_a + \epsilon) - (\payment_b + \epsilon/3) + \payment_b = - 5/6 \epsilon < 0.
\end{align*}

This contradicts DSIC, thus we conclude $t_{1,a}(u_{-1}) > \payment_a$ cannot be true. Similarly, if $t_{1,a}(u_{-1}) < \payment_a$, the price being standard requires that $\payment_a > t_{1,a}(u_{-1}) \geq 0$ and that there exists $b \in \alts$ s.t. $t_{1,b}(u_{-1}) \geq \payment_b$. Let $\epsilon \triangleq \payment_a - t_{1,a}(u_{-1}) > 0$, and let $u_1$ be a parallel type with willingness to pay 
\begin{align*}
	\maxwp_{1,a} &= \payment_a - \epsilon/2, \\
	\maxwp_{1,b} &= \payment_b + \epsilon / 3, \\
	\maxwp_{1,c} &= 0, ~\forall c \neq a,b.
\end{align*}
We can check that $b$ is the unique agent-maximizing alternative for agent $1$ under prices $\vec{\payment}$ thus $\HH_{\ell-1}$ implies that $t_{2,a}(u_{-2}) = \infty$, thus $a$ cannot be agent-maximizing for agent 2 thus cannot be selected in the economy $(u_1, u_2,\dots, u_n)$. However, $a$ is the unique agent-maximizing alternative for agent $1$, since $\maxwp_{1,a} - t_{1,a}(u_{-1}) = \epsilon/2$, whereas  $\maxwp_{1,b} - t_{1, b} (u_{-1}) \leq \maxwp_{1,b} - \payment_b =  \epsilon/3$ and $\maxwp_{1,c} - t_{1, c} (u_{-1}) \leq 0$ for all $c \neq a,b$. This is a contradiction, thus $t_{1,a}(u_{-1}) = \payment_a$ must hold. 

This completes the proof of this theorem.
\end{proof}

From $\GG_{n-1}$ and $\HH_{n-1}$, we know that when the dictator has a unique most preferred alternative under the fixed prices $\vec{a}$, the rest of the agents do not make any payment to the mechanism. However, when there are multiple most preferred alternatives that are tied for the dictator given the prices $\vec{\payment}$, $\HH_{n-1}$ does not specify what must happen to the payments from the rest of the agents. We may consider various tie-breaking mechanisms among the alternatives toward which the dictator is indifferent, for example another fixed price dictatorship, or some generalized weighted VCG mechanism between two alternatives that are tied (in which case the general non-quasi-linear utility domain  $\nqlSpace$ is parallel). These mechanism would still satisfy (P1)-(P5), and would charge the non-dictators some non-zero payments in the degenerate case when the dictator is indifferent.

\subsection{Relaxing IR and No Subsidy}  \label{appx:relax_ir_np}

Conditions (C1) and (C2) in Theorem~\ref{thm:dictatorship} require the utility domain to deviate very minimally from the parallel domain,  however, the negative result no longer holds if one of (P4) or (P5) is relaxed. 
(P4) IR and (P5) No subsidy require prices to be standard (Lemma~\ref{lem:aiprice_characterization}). With standard prices, the shapes of an agent's utility functions where the prices are negative, or where the utilities are below $\min_{a\in\alts}v_{i,a}$, are irrelevant to which alternative is agent-maximizing for this agent. The parallel domain only requires that the utility functions in the range that is relevant to be horizontal translations of each other.

As we have shown in the proof of Lemma~\ref{lem:price_characterization}, 
an affine maximizer as the choice rule together with DSIC determine the agent-independent prices that each agent faces up to a constant (when there are no ties). The requirement that prices being standard then fully pins down the agent-independent prices. Without (P4), (P5), a local violation of ``the relative willingness to pay between any two alternatives remains constant" (e.g. $u_2^\ast$ as in Figure~\ref{fig:exmp_nqltype_u2_ast}) can be made irrelevant by setting non-standard prices carefully and we still get an incentive compatible mechanism.

Here we provide a family of such mechanisms that violates only one of (P4) or (P5). This shows that the parallel domain as we defined in Section~\ref{sec:PD} is not the only maximal utility domain where mechanisms that satisfy (P1)-(P3) exist. 
We say a utility domain is parallel w.r.t. price $\payment^\ast$ if the relative willingness to pay remains the same in the range where (I) payments are at least $\payment^\ast$, and (II) the utilities are weakly above $\min_{a' \in \alts} \{ u_{i,a'}(\payment^\ast) \}$, as illustrated in Figure~\ref{fig:example_pd_z}. Denote $\maxwp_{i,a}(\payment^\ast) \triangleq u_{i,a}^{-1}(\min_{a' \in \alts} \{ u_{i,a'}(\payment^\ast)\})$.

\begin{figure}[t!]

\centering     

\begin{tikzpicture}[scale = 0.9][font = \normalsize]
\draw[->] (-0.7,0) -- (8.2,0) node[anchor=north] {$\payment$};

\draw[->] (-0.5,-0.4) -- (-0.5, 3.8) node[anchor=west] {$u_{i,a'}(\payment)$};


\draw  [-]  (-0.7, 3.7) to [out=-45, in=-200] (5, 0.8) to[out=-20, in=-190] (7.5, 0.2);
\draw  [dashed]   (-0.7, 2.35) to[out=-25, in=-205]  (0, 1.97) to[out=-25, in=-200] (3.2, 0.8) to[out=-20, in=-190] (5.5, -0.2);
\draw[dashdotted]  (-0.7, 1.6) parabola[bend at end] (2, -0.2);

\draw[dotted](0, 3.2) -- (0, -0.2) node[anchor = north] {$\payment^\ast$};

\draw[dotted](-0.5, 0.8) -- (3.2, 0.8) ;
\draw[dotted](5, 0.8) -- (5.5, 0.8) ;
\draw [dotted](3.2, 1.1) -- (3.2, -0.3);
\draw (3.2, -0.2) node[anchor=north] { { \normalsize $\maxwp_{i,b}(\payment^\ast)$ }};

\draw [dotted](5, 1.1) -- (5, -0.3);
\draw (5, -0.2) node[anchor=north] { { \normalsize $\maxwp_{i,a}(\payment^\ast)$ }};

\draw (-0.55, 3.5)node[anchor = east] {$v_{i,a}$};
\draw (-0.55, 2.1)node[anchor = east] {$v_{i,b}$};
\draw (-0.4, 1.4)node[anchor = east] {$v_{i, \worstalt_i}$};
\draw (-0.4, 0.6)node[anchor = east] {$u_{i, \worstalt_i}{\tiny (\payment^\ast)}$};

\draw[-] (6, 3) -- (6.5, 3) node[anchor=west] {$u_{i,\hspace{0.1em}a}\hspace{0.1em}(\payment)$};
\draw[dashed] (6, 2.4) -- (6.5, 2.4) node[anchor=west] {$u_{\hspace{0.1em}i,\hspace{0.1em}b\hspace{0.1em}}(\payment)$};
\draw[dashdotted] (6, 1.8) -- (6.5, 1.8) node[anchor=west] {$u_{i, \worstalt_i} (\payment)$};

\fill [pattern = north east lines, pattern color = black!20](0, 1.97) to[out=-25, in=-200] (3.2, 0.8) -- (5, 0.8) to[out=-25, in=-200](1.8, 1.97) -- (0, 1.97) ;	

\draw[dotted, thick,<->](0, 1.97) -- (1.8, 1.97) ;
\draw[dotted, thick,<->](1.8, 1.3) -- (3.6, 1.3) ;
\draw[dotted, thick,<->](3.2, 0.8) -- (5, 0.8) ;

\end{tikzpicture}
\caption{An example type in the parallel domain w.r.t. $\payment^\ast$. \label{fig:example_pd_z}}
\end{figure}
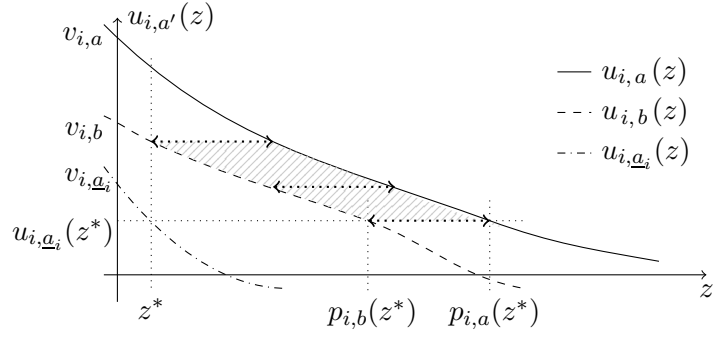

\begin{definition} A utility domain for an agent $\utilDomain_i \subset \nqlDomain$ is \emph{parallel with respect to $\payment^\ast \in \setR$} if $\forall u_i \in \utilDomain_i$, $\forall a, b \in \alts$ s.t. $u_{i,a}(\payment^\ast) \geq u_{i,b}(\payment^\ast)$, and $\forall \payment \in [\payment^\ast, ~\maxwp_{i,b}(\payment^\ast)]$, $u_{i,a} \left( \payment + (\maxwp_{i,a} (\payment^\ast)   - \maxwp_{i,b} (\payment^\ast) \right) = u_{i,b}(\payment).$
\end{definition}

\begin{definition} The generalized weighted VCG mechanism with fixed payment $\payment^\ast$ parametrized by $\payment^\ast$, non-negative weights $\{ \affcoeff_i \}_{i \in N}$ and real constants $\{ \affcnst_a \}_{a \in \alts}$ collects type profile $\hat{u}$ from the agents, and computes the willingness to pay $ \{ \hat{\maxwp}_{i,a}(\payment^\ast) \}_{i \in N, a \in \alts}$ w.r.t. $\payment^\ast$.
\begin{itemize}
	\item Choice rule: $x(\hat{u})\hsq\hsq = a^\ast$ where $a^\ast \hsq \hsq \in \hsq  \hsq \arg\max_{a \in \alts} \left\lbrace \sum_{i \in N} \affcoeff_i \hat{\maxwp}_{i,a}(\payment^\ast)\hsq  \hsq + \hsq  \affcnst_a \right\rbrace$, breaking ties arbitrarily.
	\item Payment rule: $t_i(\hat{u}) = \payment^\ast$ for $i \in N$ s.t. $\affcoeff_i = 0$; for $i$ s.t. $\affcoeff_i \neq 0$:
	\begin{align*}
		t_i(\hat{u}) \hsq \hsq = \hsq \hsq \frac{1}{\affcoeff_i} \hsq \hsq \left( \sum_{j \neq i} \affcoeff_{j} \hat{\maxwp}_{j, a^\ast_{-i}}(\payment^\ast)  \hsq + \hsq \affcnst_{a^\ast_{-i}}  \hsq \hsq -  \hsq \hsq \sum_{j \neq i} \affcoeff_{j} \hat{\maxwp}_{j, a^\ast}(\payment^\ast) \hsq - \hsq \affcnst_{a^\ast} \right) \hsq \hsq \hsq + \hsq \payment^\ast, 
 	\end{align*}
 	where 
	$a^\ast_{-i} \in \arg\max_{a \in \alts}  \{ \sum_{j \neq i} \affcoeff_{j} \hat{\maxwp}_{j,a}(\payment^\ast) + \affcnst_a \}$.
\end{itemize}
\end{definition}

\begin{proposition} 
Assuming $\utilDomain_i$ is parallel w.r.t. $\payment^\ast$ for all $i\in N$, the generalized weighted VCG mechanism with fixed payment $\payment^\ast$ and at least one non-zero coefficient $\affcoeff_i$ satisfies (P1)-(P3). In addition, (P4) is satisfied if $\payment^\ast$ is non-positive, whereas (P5) is satisfied if $\payment^\ast$ is non-negative. 
\end{proposition}

\begin{proof}  
It is immediate that the mechanism is DSIC for agents $i \in N$ s.t. $\affcoeff_i = 0$. Fix an agent $i \in N$ s.t. $\affcoeff_i > 0$. Her agent-independent prices is given by:
 $
 t_{i,a}(\hat{u}_{-i}) =  1/{\affcoeff_i} ( \sum_{j \neq i} \affcoeff_{j} \hat{\maxwp}_{j, a^\ast_{-i}}(\payment^\ast) + \affcnst_{a^\ast_{-i}} - \sum_{j \neq i} \affcoeff_{j} \hat{\maxwp}_{j, a}(\payment^\ast) - \affcnst_{a}) + \payment^\ast,$ which implies $t_{i,a^\ast_{-i}}(\hat{u}_{-i}) = \payment^\ast$ and $t_{i,a}(\hat{u}_{-i}) \geq \payment^\ast$ for all $a$. Similar to Lemma~\ref{lem:aiprice_characterization}, we can show that the agent-maximizing alternative for agent $i$ given $\{t_{i,a}(\hat{u}_{-i})\}_{a \in \alts}$ is the maximizer of $\hat{\maxwp}_{i,a}(\payment^\ast) - t_{i,a}(\hat{u}_{-i})$. We can then examine that $a^\ast$ is a maximizer of $\hat{\maxwp}_{i,a}(\payment^\ast) - t_{i,a}(\hat{u}_{-i})$, and this implies that the choice rule is agent-maximizing and this completes the proof of DSIC.

When $\payment^\ast \leq 0$, the minimum agent-independent price for  agent $i$ among all alternatives is non-positive, thus the agent-maximizing alternative gives the agent a utility at least $\min_{a \in \alts} \hat{u}_{i,a}(\payment^\ast) \geq \min_{a \in \alts} \hat{u}_{i,a}(0)$ thus (P4) IR is satisfied.
When $\payment^\ast \geq 0$, all agent-independent prices are non-negative, thus the mechanism satisfies (P5) No subsidy.
\end{proof}


\subsubsection{Proof of Theorem~\ref{thm:two_slope_dictatorship}} \label{appx:two_slope_dictatorship}

\ThmTwoSlopeDic*

We first define \emph{strictly parallel domains} as sets of agent types where 
utility curves are horizontal shifts of each other \emph{everywhere}. 
Quasi-linear utility functions, for example, are strictly parallel.

\begin{definition}[Strictly Parallel Domain] 
A utility domain $\utilDomain_i \subset \nqlDomain$ is a \emph{strictly parallel domain} if for all $u_i \in \utilDomain_i$, 
\begin{align}
	u_{i,a} \left( \payment + (\maxwp_{i,a} - \maxwp_{i,b}) \right) = u_{i,b}(\payment), ~\forall \payment \in \setR,~\forall a, b \in \alts. \label{equ:strictly_parallel_defn}
\end{align}
\end{definition}

On any strictly parallel domain, $\maxwp_{i,a} - t_{i,a} \geq \maxwp_{i,b} - t_{i,b} \Leftrightarrow u_{i,a}(t_{i,a}) \geq u_{i,b}(t_{i,b})$ holds for any prices $t_{i,a}, t_{i,b} \in \setR$, therefore, $\arg\max_{a \in \alts}  \{ u_{i,a}(t_{i,a}) \} = \arg \max_{a \in \alts} \{ \maxwp_{i,a} - t_{i,a} \}$ holds for any prices, without the requirement that the prices are standard. This implies that W-Mon is a necessary condition for any mechanism that is deterministic and DSIC for strictly parallel utility domains, and that Roberts' theorem can be generalized without the additional assumptions (P4) and (P5).

\begin{lemma} \label{lem:wmon_ic_strictlyPL} With any strictly parallel utility domain $\utilSpace$, every social choice mechanism that is DSIC and deterministic must satisfy W-Mon in terms of willingness to pay.
\end{lemma}

\begin{lemma}[Roberts' Theorem on Strictly Parallel Domains] \label{lem:gen_Roberts_SPD}
With three or more alternatives, and an unconstrained strictly parallel domain $\utilSpace$, for every social choice mechanism that satisfies (P1)-(P3), there exist non-negative weights $\affcoeff_1, \dots, \affcoeff_n$ (not all of them zero) and constants $\affcnst_1, \dots, \affcnst_m$ such that for all $u \in \utilSpace$, $x(u) \in \arg \max_{a \in \alts} \left\lbrace \sum_{i=1}^n \affcoeff_i \maxwp_{i,a} + \affcnst_a  \right\rbrace$.
\end{lemma}

We now prove the impossibility result for the linear domain with two slopes.

\begin{proof}[Proof of Theorem~\ref{thm:two_slope_dictatorship}]

Let $\utilDomain_i^\alpha \subset \utilDomain_i$ be the set of all $u_i \in \utilDomain_i$ s.t. $u_{i,a}(\payment) = v_{i,a} - \alpha_i \payment$ holds for all $a \in \alts$, i.e. the set of ``$\alpha$ types". We know $\utilDomain_i^\alpha$ is a strictly parallel domain with unrestricted willingness to pay. Let $\utilSpace^\alpha \triangleq \prod_{i=1}^n \utilDomain_i^\alpha$, Lemma~\ref{lem:subdomain_onto} implies that fixing any mechanism $(x,t)$ on $\utilSpace$ that satisfies (P1)-(P3), the restriction of $(x,t)$ on $\utilSpace^\alpha$ must also satisfy (P1)-(P3). Lemma~\ref{lem:gen_Roberts_SPD} then implies that there exists non-negative coefficients $\{\affcoeff_i\}_{i\in N}$ and real constants $\{ \affcnst_a \}_{a \in \alts}$ s.t. $k_i \neq 0$ for some $i \in N$ and $x(u) \in \arg \max_{a \in \alts} \left\lbrace \sum_{i=1}^n \affcoeff_i \maxwp_{i,a} + \affcnst_a  \right\rbrace$ for all $u \in \utilSpace^\alpha$.
With the same arguments as in the proof of Lemma~\ref{lem:price_characterization}, we can show that $\forall i \in N$ s.t. $\affcoeff_i > 0$, $\forall a,b \in \alts$, the difference in the agent-independent prices satisfies:
\begin{align}
	t_{i,a}(u_{-i}) - t_{i,b}(u_{-i}) = \frac{1}{\affcoeff_i} \left( \sum_{j \neq i} \affcoeff_{j} \maxwp_{j,b} + \affcnst_{b} -  \sum_{j\neq i} \affcoeff_{j} \maxwp_{j,a} - \affcnst_{a} \right). \label{equ:price_diff_appx}
\end{align}

Assume without loss of generality that $\affcoeff_1 > 0$, we prove that agent 1 is the fixed-price dictator with the following steps. 

\begin{itemize}
	\item Step 1. $\affcoeff_i = 0$ for all $i \neq 1$. 
	\item Step 2. There exist fixed prices $ \vec{\payment} \in \setR^m$ s.t. $\forall u_{-1} \in \utilSpace_{-1}^\alpha$, $\forall a \in \alts$, $t_{1,a}(u_{-1}) = \payment_a$.
	\item Step 3. Agent 1 is the fixed-price dictator for any $u \in \utilSpace$.
\end{itemize}

The reason that Step~2 is not immediately implied by Step~1 is that without the assumptions of (P4) and (P5), so that prices are not necessarily standard, \eqref{equ:price_diff_appx} only pins down the agent-independent prices up to a constant, and we need to show by induction that the prices $t_{1,a}(u_{-1})$ must be fixed for all $u_{-1} \in \utilSpace_{-1}^\alpha$. We then prove in Step~3 by induction that $t_{1,a}(u_{-1}) = \payment_a$ must hold for all $u_{-1} \in \utilSpace_{-1}$, which shows that agent $1$ is the fixed price dictator.

\bigskip

\noindent{}\emph{Step 1.} 
We show by contradiction that $\affcoeff_2 = 0$ must hold. The same argument can be repeated for all $i \neq 1$. We first prove the following claim. 

\begin{claim} \label{clm:wtp_diff_fixed_price} Fix any $u_{-1,-2} \in \utilSpace_{-1,-2}^\alpha$, and assume $\affcoeff_1, \affcoeff_2 > 0$. $\forall a,b\in \alts$, and $\forall u_{2}, u_{2}' \in \utilDomain_2^\alpha$, $\maxwp_{2,a} - \maxwp_{2,b} = \maxwp_{2,a}' - \maxwp_{2,b}' \Rightarrow t_{1,a}(u_2, u_{-1,-2}) = t_{1,a}(u_2', u_{-1,-2}).$
\end{claim} 

\begin{proof}

Let $u_2, u_2' \in \utilDomain_2^\alpha$ be two types of agent $2$ s.t. $\maxwp_{2,a} - \maxwp_{2,b}  = \maxwp_{2,a}' - \maxwp_{2,b}'$, and assume for contradiction that $ t_{1,a}(u_2, u_{-1,-2}) \neq t_{1,a}(u_2', u_{-1,-2})$. 
Denote $u_{-1} \triangleq (u_2, u_{-1,-2})$ and $u_{-1}' \triangleq (u_2', u_{-1,-2})$. $t_{1,a}(u_{-1}) - t_{1,b}(u_{-1}) = t_{1,a}(u_{-1}')- t_{1,b}(u_{-1}')$ follows from \eqref{equ:price_diff_appx}. 
Assume w.l.o.g. that $t_{1,a}(u_{-1}) < t_{1,a}(u_{-1}')$, and that $\alpha_1 > \beta_1$. 
We can construct $u_1 \in \utilDomain_1$ as illustrated in Figure~\ref{fig:TwoSlopeDictator_Step1} such that $u_{1,a}(\payment) = v_{1,a} - \alpha_1 \payment$, $u_{1,b}(\payment) = v_{1,b} - \beta_1 \payment$, and that
\begin{align}
		u_{1,c}(t_{1,c}(u_{-1})) & < u_{1,b}(t_{1,b}(u_{-1})) < u_{1,a}(t_{1,a}(u_{-1})), ~\forall c \neq a,b, \label{equ:utility_inequ} \\
		u_{1,c}(t_{1,c}(u_{-1}')) & < u_{1,a}(t_{1,a}(u_{-1}'))  < u_{1,b}(t_{1,b}(u_{-1}')), ~\forall c \neq a,b. \label{equ:utility_inequ_prime}
\end{align}

\begin{figure}[t!]
\vspace{-0.0em}
\centering  

\begin{tikzpicture}[scale = 0.95][font = \normalsize]
\draw[ ->] (-0.1,0) -- (8,0) node[anchor=north] {$\payment$};

\draw[ ->] (0,-0.2) -- (0, 3.7) node[anchor=west] {$u_{1,a'}(\payment)$};

\draw[line width=0.2mm, -] (-0.2, 2.2) -- (2.2, -0.2);

\draw  [line width=0.2mm, dashed]  (-0.2, 3.2) -- (7, -0.2); 

\draw (0, 3.1)node[anchor = east] {{$v_{1,b}$}};
\draw (0, 2)node[anchor = east] {{$v_{1,a}$}};

\draw[line width=0.2mm, -] (6.2, 3.5) -- (6.85, 3.5) node[anchor=west] {$u_{1,a}(\payment)$};
\draw[line width=0.2mm, dashed] (6.2, 2.75) -- (6.85, 2.75) node[anchor=west] {$u_{1,b}(\payment)$};

\draw[dotted](0, 1.6) -- (3.8, 1.6) ;
\draw[dotted](0, 0.4) -- (6, 0.4) ;

\draw[dotted](0.35, 2) -- (0.35, -0.2);
\draw (0.1, -0.1) node[anchor = north]{ { $t_{1,a}(u_{-1})$ }} ;

\draw[dotted](1.6, 0.7) -- (1.6, -0.2);
\draw (1.85, -0.1) node[anchor = north]{ { $t_{1,a}(u_{-1}')$ }} ;

\draw[dotted](3.85, 1.7) -- (3.85, -0.2);
\draw (3.65, -0.1) node[anchor = north]{ { $t_{1,b}(u_{-1})$ }} ;

\draw[dotted](5.1, 0.9) -- (5.1, -0.2);
\draw (5.4, -0.1) node[anchor = north]{ { $t_{1,b}(u_{-1}')$ }} ;

\end{tikzpicture}
\caption{Illustration of $u_1 \in \utilDomain_1 \backslash \utilDomain_1^\alpha$, for the proof of Claim~\ref{clm:wtp_diff_fixed_price}.
\label{fig:TwoSlopeDictator_Step1}}
\end{figure}

Such $u_1$ exists since $\alpha_1 > \beta_1$, 	$t_{1,a}(u_{-1}) - t_{1,b}(u_{-1}) = t_{1,a}(u_{-1}') - t_{1,b}(u_{-1}')$, and the values are unrestricted. Agent-maximizing for agent $1$ then implies $x(u_1, u_{-1}) = a$ and $x(u_1, u_{-1}') = b$. Note that the agent-independent prices facing agent 2, $\{t_{2,a'}(u_{-2}) \}_{a' \in \alts}$, remain the same in the two economies. Agent-maximization for agent 2 therefore implies
\begin{align*}
	\maxwp_{2,a} \hsq - t_{2,a}(u_{-2}) \hsq \geq \maxwp_{2,b} \hsq - t_{2,b}(u_{-2}) ,\\ 
	\maxwp_{2,a}' \hsq - t_{2,a}(u_{-2}) \hsq \leq \maxwp_{2,b}' \hsq  - t_{2,b}(u_{-2}). 
\end{align*}
Given the assumption $\maxwp_{2,a} - \maxwp_{2,b} = \maxwp_{2,a}' - \maxwp_{2,b}'$, we must have
\begin{align}
	\maxwp_{2,a} -  \maxwp_{2,b} = \maxwp_{2,a}' -  \maxwp_{2,b}' = t_{2,a}(u_{-2}) - t_{2,b}(u_{-2}). \label{equ:equ_constraint}
\end{align}

%

Let $\epsilon$ be some small positive number such that $ 0 < \epsilon < u_{1,b}(t_{1,b}(u'_{-1})) - u_{1,a}(t_{1,a}(u'_{-1}))$. Consider  type $u_1' \in \utilDomain_1$ such that for all $z \in \setR$, $u'_{1,a}(\payment) = u_{1,a}(\payment) $, $u'_{1,b}(\payment) = u_{1,b}(\payment) - \epsilon$ and $u'_{1,c}(\payment) = u_{1,c}(\payment)$. In words, $u_1'$ is identical to $u_1$, except that $v_{1,b}' = v_{1,b} - \epsilon$, where $\epsilon$ is small enough that both \eqref{equ:utility_inequ} and \eqref{equ:utility_inequ_prime} still hold if $u_1$ is replaced with $u_1'$. Thus we have $x((u_1', u_2, u_{-1,-2})) = a$ and $x((u_1', u_2', u_{-1,-2})) = b$ from agent-maximization for agent $1$ given $u_1'$.

Replacing $u_1$ with $u_1'$ results in a decrease in the willingness to pay for alternative $b$, thus $\maxwp_{1,a}' - \maxwp_{1,b}' > \maxwp_{1,a} - \maxwp_{1,b}$. Given \eqref{equ:price_diff_appx} and the assumption that $\affcoeff_1, \affcoeff_2 > 0$, we know $ t_{2,a}((u_1', u_{-1,-2})) -  t_{2,b}((u_1', u_{-1,-2})) < t_{2,a}(u_{-2}) -  t_{2,b}(u_{-2})$. Combined with \eqref{equ:equ_constraint}, we know: 
\begin{align*}
	& t_{2,a}((u_1', u_{-1,-2}))  -  t_{2,b}((u_1', u_{-1,-2}))  < \maxwp_{2,a}' - \maxwp_{2,b}'  \\
	\Rightarrow & \maxwp_{2,a}' - t_{2,a}((u_1', u_{-1,-2})) >   \maxwp_{2,b}' - t_{2,b}((u_1', u_{-1,-2})),
\end{align*}
meaning that the alternative $x((u_1', u_2', u_{-1,-2})) = b$ is not agent-maximizing for agent 2 in economy $(u_1', u_2', u_{-1,-2})$. 
%
%
This contradicts DSIC, thus we conclude that $t_{1,a}(u_2, u_{-1,-2}) = t_{1,a}(u_2', u_{-1,-2})$ must hold. 
\end{proof}

Assume $\affcoeff_1, \affcoeff_2 > 0$, fix any $u_{-1,-2} \in \utilSpace_{-1,-2}^\alpha$ and some $u_2^\ast \in \utilDomain^\alpha_2$, we know $t_{1,a}(u_2^\ast, u_{-1,-2}) < \infty$, since $\affcoeff_1 > 0$ and $\utilDomain_1^\alpha$ is unrestricted thus there exists $u_1 \in \utilDomain_1^\alpha$ s.t. $x((u_1, u_2^\ast, u_{-1,-2})) = a$. Denote $\Delta_{a,b}^\ast \triangleq \maxwp_{2,a}^\ast - \maxwp_{2,b}^\ast$, $\Delta_{a,c}^\ast \triangleq \maxwp_{2,a}^\ast - \maxwp_{2,c}^\ast$ and $\payment_a \triangleq t_{1,a}(u_2^\ast, u_{-1,-2})$. 
We know from Claim~\ref{clm:wtp_diff_fixed_price} that $\forall u_2 \in \utilDomain_2^\alpha$ s.t. $\maxwp_{2,a} - \maxwp_{2,c} = \Delta_{a,c}^\ast$, $t_{1,a}(u_2, u_{-1,-2}) = \payment_a$. 
For any $\Delta_{a,b} \in \setR$, we can find $u_2' \in \utilDomain_2^\alpha$ s.t. $\maxwp_{2,a}' - \maxwp_{2,b}' = \Delta_{a,b}$ and $\maxwp_{2,a}' - \maxwp_{2,c}' = \Delta_{a,c}^\ast$, for which $t_{1,a}(u_2', u_{-1,-2}) = \payment_a$. Apply Claim~\ref{clm:wtp_diff_fixed_price} again, we know that for all $u_2 \in \utilSpace_2^\alpha$ s.t. $\maxwp_{2,a} - \maxwp_{2,b} = \Delta_{a,b}$, $t_{1,a}(u_2, u_{-1,-2}) = \payment_a$. Since this holds for all $\Delta_{a,b} \in \setR$, we conclude that $\forall u_{2} \in \utilDomain^\alpha_{2}$, $t_{1,a}(u_{2}, u_{-1,-2} ) = \payment_a$.
%
%
This implies that even for $u_2$, $u_2' \in \utilDomain_2^\alpha$ s.t. $\maxwp_{2,a} - \maxwp_{2,b} \neq \maxwp_{2,a}' - \maxwp_{2,b}' $, we still have $t_{1,a}(u_2, u_{-1,-2}) - t_{1,b}(u_2, u_{-1,-2}) = t_{1,a}(u_2', u_{-1,-2}) - t_{1,b}(u_2', u_{-1,-2} )$. This contradicts \eqref{equ:price_diff_appx}, thus $\affcoeff_2 = 0$ must hold.

\bigskip

\noindent{}\emph{Step~2.} Given \eqref{equ:price_diff_appx} and Step~1 that $\affcoeff_i = 0$ for all $i \neq 1$, we know that for all $u_{-1} \in \utilSpace_{-1}^\alpha$, for any pair of alternatives $a, b \in \alts$, $t_{1,a}(u_{-1})  - t_{1,b}(u_{-1}) = 1/\affcoeff_1(\affcnst_b - \affcnst_a)$ must hold. Denote this price difference as $\delta_{a,b} \triangleq  1/\affcoeff_1(\affcnst_b - \affcnst_a)$. 
We prove the following claim.

\begin{claim} \label{clm:diff_price_same_wp} For all $i \neq 1$, for all  $u_{-1,-i} \in \utilSpace_{-1,-i}^\alpha$, and for all $u_i, u_i' \in \utilDomain_i^\alpha$, if there exists $a^\ast \in \alts$ s.t.  $t_{1,a^\ast}((u_i', u_{-1,-i})) \neq t_{1,a^\ast}((u_i, u_{-1,-i}))$, then $\maxwp_{i,a} = \maxwp_{i,a}'$ for all $ a \in \alts$, i.e. $u_i$ and $u_i'$ have the same willingness to pay for all $a \in \alts$.
\end{claim}

\begin{proof} Fix agent $i = 2$ and any $u_{-1,-i} \in \utilSpace_{-1,-i}^\alpha$, denote $u_{-1} \triangleq　(u_2, u_{-1,-2})$, $u_{-1}' = (u_2', u_{-1,-2})$, and assume that there exists $a\in \alts$ s.t. $t_{1,a}(u_{-1}') >  t_{1,a}(u_{-1})$. We first prove $\maxwp_{2,a} - \maxwp_{2,a}' = \maxwp_{2,b} - \maxwp_{2,b}'$ for all $b \in \alts$. This pins down the willingness to pay for all alternatives up to a constant. The fact that the smallest willingness to pay among all alternatives must be zero then implies $\maxwp_{2,b} = \maxwp_{2,b}'$ for all $b \in \alts$. Repeating the same arguments for all $i \neq 1$ completes the proof of this claim.

We now prove $\maxwp_{2,a} - \maxwp_{2,a}' = \maxwp_{2,b} - \maxwp_{2,b}'$. First, we know $t_{1,a}(u_{-1}') - t_{1,b}(u_{-1}') = t_{1,a}(u_{-1}) - t_{1,b}(u_{-1}) = \delta_{a,b}$ from Step~1. 
Same as the proof of Claim~\ref{clm:wtp_diff_fixed_price}, we may find $u_1 \in \utilDomain_1$ as shown in Figure~\ref{fig:TwoSlopeDictator_Step1}, where \eqref{equ:utility_inequ} and \eqref{equ:utility_inequ_prime} both hold. Agent-maximization for agent $1$ then implies
\begin{align}
	x((u_1, u_2, u_{-1,-2})) = a, ~x((u_1, u_2', u_{-1,-2})) = b. \label{equ:choice_ab}
\end{align}
Similarly, there exists $u_1' \in \utilDomain_1$ as illustrated in Figure~\ref{fig:TwoSlopeDictator_Step2} such that   
$u_{1,a}'(\payment) = v_{1,a}' - \beta_1 \payment$, $u_{1,b}'(\payment) = v_{1,b}' - \alpha_1 \payment$, and that 
\begin{align*}
		u_{1,c}'(t_{1,c}(u_{-1})) & < u_{1,a}'(t_{1,a}(u_{-1})) < u_{1,b}'(t_{1,b}(u_{-1})), ~\forall c \neq a,b, \\
		u_{1,c}'(t_{1,c}(u_{-1}')) & < u_{1,b}'(t_{1,b}(u_{-1}'))  < u_{1,a}'(t_{1,a}(u_{-1}')), ~\forall c \neq a,b. 
\end{align*}

\begin{figure}[t!]
\vspace{-0.0em}
\centering  

\begin{tikzpicture}[scale = 0.95][font = \normalsize]
\draw[ ->] (-0.1,0) -- (8,0) node[anchor=north] {$\payment$};

\draw[ ->] (0,-0.2) -- (0, 3) node[anchor=west] {$u'_{1,a'}(\payment)$};

\draw[line width=0.2mm, -] (-0.2, 1.6) -- (3.3, -0.2);

\draw  [line width=0.2mm, dashed]  (2.6, 2.76) -- (5.7, -0.2); 

\draw (0, 1.5)node[anchor = east] {{$v'_{1,a}$}};

\draw[line width=0.2mm, -] (6, 2.8) -- (6.65, 2.8) node[anchor=west] {$u'_{1,a}(\payment)$};
\draw[line width=0.2mm, dashed] (6, 2) -- (6.65, 2) node[anchor=west] {$u'_{1,b}(\payment)$};

\draw[dotted](0, 1.6) -- (3.8, 1.6) ;
\draw[dotted](0, 0.4) -- (6, 0.4) ;

\draw[dotted](0.35, 2) -- (0.35, -0.2);
\draw (0.1, -0.1) node[anchor = north]{ { $t_{1,a}(u_{-1})$ }} ;

\draw[dotted](1.6, 0.7) -- (1.6, -0.2);
\draw (1.85, -0.1) node[anchor = north]{ { $t_{1,a}(u_{-1}')$ }} ;

\draw[dotted](3.85, 1.7) -- (3.85, -0.2);
\draw (3.65, -0.1) node[anchor = north]{ { $t_{1,b}(u_{-1})$ }} ;

\draw[dotted](5.1, 0.9) -- (5.1, -0.2);
\draw (5.4, -0.1) node[anchor = north]{ { $t_{1,b}(u_{-1}')$ }} ;

\end{tikzpicture}
\caption{Illustration of $u_1' \in \utilDomain_1$, for the proof of Claim~\ref{clm:diff_price_same_wp}.
\label{fig:TwoSlopeDictator_Step2}}
\end{figure}
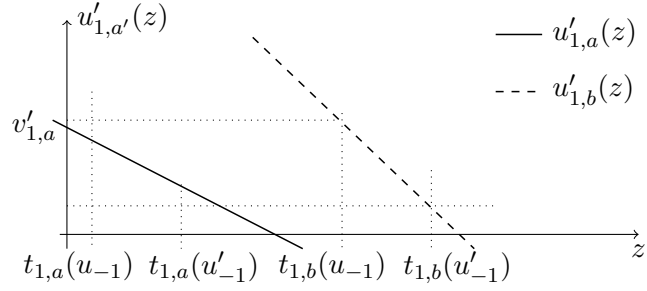
\noindent{}Agent-maximization for agent 1 with type $u_1'$ then requires
\begin{align}
	x((u_1', u_2, u_{-1,-2})) = b, ~x((u_1', u_2', u_{-1,-2})) = a. \label{equ:choice_ba}
\end{align}
Given \eqref{equ:choice_ab}, we know from W-Mon that $\maxwp_{2,a} - \maxwp_{2,a}' \geq \maxwp_{2,b} - \maxwp_{2,b}'$ must hold. Similarly, we get $\maxwp_{2,a} - \maxwp_{2,a}' \leq \maxwp_{2,b} - \maxwp_{2,b}'$ from \eqref{equ:choice_ba}, therefore $\maxwp_{2,a} - \maxwp_{2,a}' = \maxwp_{2,b} - \maxwp_{2,b}'$.
\end{proof}

Intuitively, Claim~\ref{clm:diff_price_same_wp} shows that for any $i \neq 1$, for any $u_{-1,-i} \in \utilSpace_{-1,-i}^\alpha$, and any all $u_i, u_i' \in \utilDomain_i^\alpha$, if there exists $a \in \alts$ s.t. $\maxwp_{i,a} \neq \maxwp_{i,a}'$, then we must have $t_{1,a'}((u_i', u_{-1,-i})) = t_{1,a'}((u_i, u_{-1,-i}))$ for all $a' \in \alts$.  
Fix any $u_{-1}^\ast = (u_2^\ast, \dots, u_n^\ast) \in \utilSpace_{-1}^\alpha$, and define $\payment_a \triangleq t_{1,a}(u_{-1}^\ast)$ for all $a \in \alts$. 
We prove by induction that $\forall u_{-1} \in \utilSpace_{-1}^\alpha$, $t_{1,a}(u_{-1}) = \payment_a$ must hold for all $a \in \alts$.
For any $\ell = 0, 1, \dots, n-1$, let the induction statements be 
\begin{enumerate}
	\item[$\GG_\ell^\ast$:] $\forall u_{-1} \in \utilSpace_{-1}^\alpha$ s.t. $|\{ i \in N | i \neq 1, u_i \neq u_i^\ast \}|\leq \ell$, $t_{1,a}(u_{-1}) = \payment_{a}, \forall a \in \alts$.
	\item [$\HH_\ell^\ast$:] $\forall i \neq 1$, for all $u_1 \in \utilDomain_1$ and $\forall u_{-1,-i} \in \utilSpace_{-1,-i}^\alpha$, if (I) $|\{ j \in N | j \neq 1, j \neq i, u_j \neq u_j^\ast \}|\leq \ell$, and (II) $\exists a^\ast$ s.t. $u_{1,a^\ast}(\payment_{a^\ast}) > u_{1,a}(\payment_a)$, $\forall a \neq a^\ast$, then $t_{i,a}(u_{-i}) = +\infty$ $\forall a \neq a^\ast$.
\end{enumerate}

$\GG_0^\ast$ trivially holds from agent-independence since when  $|\{ i \in N | i \neq 1, u_i \neq u_i^\ast \}|\leq 0$, $u_{-1} = u_{-1}^\ast$. We now prove $\GG_\ell^\ast \Rightarrow \HH_\ell^\ast$ and $\HH_{\ell-1}^\ast \Rightarrow \GG_{\ell}^\ast$, and this would complete the proof of the claim of this step, which is $\GG_{n-1}^\ast$. 

\smallskip

\noindent{}\emph{Step 2.1. $\GG_\ell^\ast \Rightarrow \HH_\ell^\ast$ for all $0 \leq \ell \leq n-1$.} 

Consider agent $i = 2$, and fix any $u_{-1, -2} \in \utilSpace_{-1, -2}^\alpha$ s.t. $|\{ j \in N | j \neq 1, j \neq 2, u_j \neq u_j^\ast \}|\leq \ell$. We know from $\GG_\ell^\ast$ that for all $a \in \alts$, $t_{1, a}(u_2^\ast, u_{-1, -2}) = \payment_a$.
Fix some $u_1 \in \utilDomain_1$ for which there exists $a^\ast \in \alts$ s.t. $u_{1,a^\ast}(\payment_{a^\ast}) > u_{1,a}(\payment_a)$ for all $a \neq a^\ast$. 
For any alternative $b \neq a^\ast$, and any $u_2 \in \utilDomain_2^\alpha$ s.t. $\maxwp_{2,b} > \maxwp_{2,b}^\ast$, we know from Claim~\ref{clm:diff_price_same_wp} that $t_{1,a}(u_2, u_{-1,-2}) =t_{1,a}(u_2^\ast, u_{-1,-2}) =  \payment_a$ for all $a \in \alts$, thus $a^\ast$ is also the unique agent-maximizing alternative for agent $1$ in the economy $(u_1, u_2, u_{-1,-2})$. Since $\maxwp_{2,b}$ can be arbitrarily large, we must have $t_{2,b}(u_1, u_{-1,-2}) = \infty$ so that $a^\ast$ is also always agent-maximizing for agent $2$. The same argument can be repeated for all agents $i\neq 1$. 

\smallskip

\noindent{}\emph{Step 2.2. $\GG_{\ell-1}^\ast \txtand \HH_{\ell-1}^\ast \Rightarrow \GG_{\ell}^\ast$, for all $1 \leq \ell \leq n-1$.} 

Consider any type profile $u_{-1} \in \utilSpace_{-1}^\alpha$ such that $|\{ i \in N | i \neq 1, u_i \neq u_i^\ast \}| = \ell$. We assume w.l.o.g. that $u_2 \neq u_2^\ast$. We know from $\GG_{\ell-1}^\ast$ that $t_{1,a}(u_{2}^\ast, u_{-1,-2}) = \payment_a$ holds for all $a \in \alts$, since in $(u_2^\ast, u_{-1,-2})$, $|\{ i \in \alts | i \neq 1, u_i \neq u_i^\ast \}| = \ell - 1$. Given that both $u_2$ and $u_2^\ast$ are in $\utilDomain_2^\alpha$, if there exists any alternative $a \in \alts$ s.t. $\maxwp_{2,a} \neq \maxwp_{2,a}^\ast$, we know from Claim~\ref{clm:diff_price_same_wp} that $t_{1,a}(u_2, u_{-1,-2}) = t_{1,a}(u_2^\ast, u_{-1,-2}) = \payment_a$ holds for all $a \in \alts$, which is what we are looking for. 
Therefore, the only remaining case is for $u_2 \in \utilDomain_2^\alpha$ s.t. $\maxwp_{2,a} = \maxwp_{2,a}^\ast$ for all $a \in \alts$.

Assume toward a contradiction, that there exists $u_2 \in \utilDomain_2^\alpha$ s.t. $\maxwp_{2,a} = \maxwp_{2,a}^\ast$ for all $a \in \alts$, for which there exists an alternative, say alternative $a$ s.t. $t_{1,a}(u_2, u_{-1,-2}) \neq t_{1,a}(u_2^\ast, u_{-1,-2}) = \payment_a$. Assume $t_{1,a}(u_2, u_{-1,-2}) > t_{1,a}(u_2^\ast, u_{-1,-2})$, and the other direction can be proved similarly. Fix any alternative $b \neq a$, we know from \eqref{equ:price_diff_appx} that $t_{1,a}(u_2, u_{-1,-2}) - t_{1,b}(u_2, u_{-1,-2}) = t_{1,a}(u_2^\ast, u_{-1,-2}) - t_{1,b}(u_2^\ast, u_{-1,-2})$ must hold. Similar to the proof of Claim~\ref{clm:wtp_diff_fixed_price}, we can find $u_1 \in \utilDomain_1$ s.t. $u_{1,a}(\payment) = v_{1,a} - \alpha_1 \payment$, $u_{1,b} = v_{1,b} - \beta_1 \payment$ for all $z \in \setR$, and that
\begin{align*}
		u_{1,c}(t_{1,c}(u_2^\ast, u_{-1,-2})) & < u_{1,b}(t_{1,b}(u_2^\ast, u_{-1,-2})) < u_{1,a}(t_{1,a}(u_2^\ast, u_{-1,-2})), ~\forall c \neq a,b,  \\
		u_{1,c}(t_{1,c}(u_2, u_{-1,-2})) & < u_{1,a}(t_{1,a}(u_2, u_{-1,-2}))  < u_{1,b}(t_{1,b}(u_2, u_{-1,-2})), ~\forall c \neq a,b.
\end{align*}

We know that in the economy $(u_1, u_2, u_{-1,-2})$, alternative $b$ is the unique agent-maximizing alternative for agent $1$. However, given since $t_{1,a'}(u_2^\ast, u_{-1,-2}) = \payment_{a'}$ holds for all $a' \in \alts$, we know that given the vector of prices $\vec{\payment}$, alternative $a$ is the unique agent-maximizing alternative for agent $1$. 
Therefore, $\HH_{\ell-1}^\ast$ together with $|\{ i \in \alts | i \neq 1, i\neq 2, u_i \neq u_i^\ast \}| = \ell - 1$ implies that $t_{2,b}(u_1, u_{-1,-2}) = \infty$. This shows that alternative $b$ cannot be agent-maximizing for agent 2 in the economy $(u_1, u_2, u_{-1,-2})$. This violates DSIC, completes the proof of $\GG_{\ell-1}^\ast \txtand \HH_{\ell-1}^\ast \Rightarrow \GG_{\ell}^\ast$, and also the proof of Step~2 of this theorem.

\bigskip

\noindent{}\emph{Step~3.} Step~2 implies that when $u_{-1} \in \utilSpace^\alpha_{-1}$, the outcome of the mechanism must be determined according to the fixed price dictator where agent $1$ is the dictator and the fixed prices are given by $\vec{\payment}$. 
The proof of the third step is very similar to the proof of Step~2 of Theorem~\ref{thm:dictatorship}: by induction on the number of agents whose type is outside of $\utilDomain_i^\alpha$, we can show that the outcome must be determined by the same fixed-price dictatorship for any $u \in \utilSpace$. This completes the proof of Theorem~\ref{thm:two_slope_dictatorship}.
\end{proof}

\section{Impossibility Result Assuming Pareto Efficiency} \label{appx:result_with_PE}

We show in this section that if condition (P3) onto is replaced by a stronger Pareto efficiency (PE), we obtain a similar dictatorship result as in Theorem~\ref{thm:dictatorial_with_PE} with weaker assumptions on the utility domain: instead of requiring all but one agent having non-parallel types, it is sufficient that there exists a single agent with a single non-parallel type.
Moreover, these efficient mechanisms must be full dictatorships instead of fixed-price dictatorships--- the dictator chooses her favorite alternative, and efficiency requires that she does not make any payment to the mechanism.

\begin{definition}[Pareto Efficiency] \label{defn:pe} A social choice mechanism 
is \emph{Pareto efficient} if for any type profile, 
it is impossible to select another alternative and adjust the payments in a way that (i) all agents are weakly better off, (ii) at least one agent is strictly better off, and (iii) the total payment to the mechanism is not reduced.
\end{definition}

\begin{theorem} \label{thm:dictatorial_with_PE} With at least three alternatives and a utility domain $\utilSpace = \prod_{i=1}^n \utilDomain_i$ such that  
\begin{enumerate}
	\item[(C1)] for each $i\in N$, $\utilDomain_i$ contains an unrestricted parallel domain,
	\item[(C2')] there exists an agent $i \in N$ such that $\utilDomain_i \not\subset \plDomain$,
\end{enumerate}
the only social choice mechanisms that satisfies (P1) DSIC, (P2) deterministic, (P4) IR, (P5) no subsidy, and Pareto efficiency, are dictatorships.
\end{theorem}

\begin{proof} 
We first observe that any mechanism that satisfy the given conditions must also be onto, therefore all (P1)-(P5) hold. For any alternative $a \in \alts$, consider the parallel type profile $u$ where the willingness to pay is given by
\begin{align*}
	\maxwp_{1,a} &= 1;~\maxwp_{1,a'} = 0, ~\forall a' \neq a,\\
	\maxwp_{i,a'} &= 0, ~\forall a' \in \alts. 
\end{align*}
Such a type profile is guaranteed to exist given (C1). Assume that any alternative $b \neq a$ is selected, we know from IR and No subsidy that there cannot be any payment from any agent. However, this violates PE, since if alternative $a$ is selected instead without charging any payment from any agent, agent $1$ is strictly better off while the rest of the agents are indifferent. This shows that every alternative is selected for some type profile thus the mechanism is onto.

\medskip

Following the notation as in the proof of Theorem~\ref{thm:dictatorship}, we denote the parallel subdomains as $\udomainintpl_i \triangleq \utilDomain_i \cap \plDomain$ for each $i \in N$, and let $\uspaceintpl \triangleq \utilSpace \cap \plSpace$ be the set of all parallel type profiles. Theorem~\ref{thm:roberts} implies any mechanism that satisfy (P1)-(P5) must be an affine maximizer on $\uspaceintpl$, with non-negative coefficients $\{\affcoeff_i\}_{i \in N}$, not all of them are zero, and real constants $\{\affcnst_a\}_{a \in \alts}$.

We first show that Pareto efficiency requires that the constants have to be all equal. Let $a \in \arg \max_{a' \in \alts} \affcnst_{a'}$. Assume that the constants are not all equal, there exists $b \in \alts$ s.t. $\affcnst_a > \affcnst_b$. Since there is at least one agent with $\affcoeff_i > 0$, assume w.l.o.g. that $\affcoeff_1 > 0$. Denote $\epsilon \triangleq 1/\affcoeff_1 (\affcnst_a - \affcnst_b)$, and consider a parallel type profile with willingness to pay given by:
\begin{align*}
	\maxwp_{1,a} & = \epsilon / 2, ~\maxwp_{1,b} = \epsilon,~\maxwp_{1,c}  = 0, ~\forall c \neq a,b, \\
	\maxwp_{i,a'} &= 0, ~\forall i \neq 1, ~\forall a' \in \alts.
\end{align*}
 
We can check that alternative $a$ is the unique maximizer of the choice rule, and therefore will be selected. Lemma~\ref{lem:price_characterization} then implies that no agent makes any payment to the mechanism, since $t_{i,a}(u_{-1}) = 1/\affcoeff_1(\affcnst_a - \affcnst_a ) = 0$, and there is no tie in the choice rule. 
This violates Pareto efficiency, since selecting alternative $b$ without charging any payment from any agent is a strict improvement for agent $1$, and does not make any other agent worse off. This implies that the constants have to be all equal. 

\medskip

We will show that the number of agents with a non-zero coefficient is exactly one, similar to Step~1 of the proof of Theorem~\ref{thm:dictatorship}. Since Step~2 of the proof of Theorem~\ref{thm:dictatorship} does not depend on the assumption (C2) and relies only on (C1), we can then repeat the same argument and reach the conclusion that the mechanism must be a fixed price dictatorship. The fact that the coefficients $\{ \affcnst_a \}_{a \in \alts}$ being all equal then implies that the fixed prices that the dictator faces must be all zero.

We now prove $|\{i \in N | \affcoeff_i > 0\}| = 1$.
Given $(C2')$, we assume w.l.o.g. that there exist a non-parallel type for agent $1$: $u_1^\ast \in \utilDomain_1 \backslash \plDomain$.
With the same arguments as in Step~1 in the proof of Theorem~\ref{thm:dictatorship}, we can show that for any other agent $i \neq 1$, there is at most one of them with a non-zero coefficient. 
This implies that either (i) $\affcoeff_1 = 0$, or (ii) $\affcoeff_i = 0$ for all $i \neq 1$. For case (ii), we already have $|\{i \in N | \affcoeff_i > 0\}| = 1$. For case (ii), we only need to consider the case where there are at least two other agents with non-zero coefficients, which we name agent $2$ and $3$. Consider a parallel type profile with willingness to pay given by:
\begin{align*}
	\maxwp_{1,a} &= 10 / \affcoeff_2, ~\maxwp_{1,b} = 0,~\maxwp_{1,c} = 0, \\
	\maxwp_{2,a} &= 5/\affcoeff_2, ~\maxwp_{2,b} = 10 / \affcoeff_2,~\maxwp_{2,c} = 0, \\
	\maxwp_{3,a} &= 4/\affcoeff_3,~ \maxwp_{3,b} = 0, ~\maxwp_{3,c} = 8/\affcoeff_3,
\end{align*}
and if there are any other agent, or any alternative other than $a,b,c$, let the willingness to pay be all zeros. Since the constants are all equal, and given $\affcoeff_1 = 0$, we know $x(u) = b = \arg\max_{a' \in \alts} \{  \affcoeff_2 \maxwp_{2,a'} + \affcoeff_3 \maxwp_{3, a'} \}$, and Lemma~\ref{lem:price_characterization} again implies that $t_2(u) = 8/\affcoeff_2$, and $t_i(u) = 0$ for all $i \neq 2$. This is not a Pareto efficient outcome, since if we select alternative $a$ instead, collect 
$8 / \affcoeff_2$ from agent $1$, and charge no payment from any other agent, all of agents $1$, $2$ and $3$ are strictly better off, and the rest of the agents, if there are any, are indifferent. This is a contradiction, implies that $|\{i \in N | \affcoeff_i > 0\}| = 1$ must hold, and completes the proof of the theorem.
\end{proof}

\end{document}